\documentclass[onecolumn,draftclsnofoot,12pt]{IEEEtran}

\usepackage{acronym}
\usepackage{amsfonts}
\usepackage[dvips]{graphicx}
\usepackage{times}
\usepackage{cite}
\usepackage{amsmath}
\usepackage{array}
\usepackage{amssymb}
\usepackage{stfloats}
\usepackage{diagbox}
\usepackage{graphicx}
\usepackage{footnote}
\usepackage{amsthm}
\usepackage{booktabs}
\usepackage{array}
\usepackage{algorithmic}
\usepackage{algorithm}
\usepackage{subeqnarray}
\usepackage{cases}
\usepackage{threeparttable}
\usepackage{color}
\usepackage{epstopdf}
\usepackage{multirow}
\usepackage{tabularx}
\usepackage{enumerate}
\usepackage{multicol}
\usepackage{hyperref}
\usepackage{subcaption}

\newtheorem{corollary}{Corollary}

\newtheorem{remark}{Remark}
\newtheorem{proposition}{Proposition}

\newcommand{\figref}[1]{{Fig.}~\ref{#1}}

\def\bb0{{\mathbb{0}}}

\def\ba{{\mathbf{a}}}
\def\bb{{\mathbf{b}}}

\def\bn{{\mathbf{n}}}

\def\bp{{\mathbf{p}}}

\def\br{{\mathbf{r}}}
\def\bs{{\mathbf{s}}}

\def\bv{{\mathbf{v}}}

\def\bx{{\mathbf{x}}}

\def\b0{{\mathbf{0}}}

\def\bA{{\mathbf{A}}}

\def\bF{{\mathbf{F}}}
\def\bG{{\mathbf{G}}}
\def\bH{{\mathbf{H}}}
\def\bI{{\mathbf{I}}}

\def\bQ{{\mathbf{Q}}}

\def\bU{{\mathbf{U}}}
\def\bV{{\mathbf{V}}}
\def\bW{{\mathbf{W}}}

\def\cH{\mathcal{H}}

\def\sf0{{\mathsf{0}}}

\def\rmN{\mathrm{N}}

\def\rm0{{\mathrm{0}}}

\def\Nt{{N_\mathrm{t}}}

\def\Nr{{N_\mathrm{r}}}

\def\Nct{{N_\mathrm{RF}^{\mathrm{t}}}}
\def\Ncr{{N_\mathrm{RF}^{\mathrm{r}}}}
\def\Ns{{N_\mathrm{s}}}

\def\Pt{{P_{\mr{t}}}}
\def\Hb{{\mathcal{H}_{\mr{b}}}}

\newcommand{\mb}{\mathbf}
\newcommand{\mr}{\mathrm}
\def\j{\mathrm{j}}
\def\Re{\mathrm{Re}}
\def\Im{\mathrm{Im}}
\acrodef{CSI}[CSI]{channel state information}
\acrodef{CSIT}[CSIT]{channel state information at the transmitter}
\acrodef{CSIR}[CSIR]{channel state information at the receiver}
\acrodef{MIMO}[MIMO]{multiple-input multiple-output}
\acrodef{SISO}[SISO]{single-input single-output}
\acrodef{MISO}[MISO]{multiple-input single-output}
\acrodef{SIMO}[SIMO]{single-input multiple-output}
\acrodef{ADCs}[ADCs]{analog-to-digital convertors}
\acrodef{SNR}[SNR]{signal-to-noise ratio}
\acrodef{AWGN}[AWGN]{additive white Gaussian noise}
\acrodef{MRT}[MRT]{maximal ratio transmission}
\acrodef{DFT}[DFT]{Discrete Fourier Transform}
\acrodef{ULA}[ULA]{uniform linear array}
\acrodef{UPA}[UPA]{uniform planar array}
\acrodef{LS}[LS]{least squares}
\acrodef{ALMMSE}[ALMMSE]{approximate linear minimum mean squared error}
\acrodef{QIHT}[QIHT]{quantized iterative hard thresholding}
\acrodef{QIST}[QIST]{quantized iterative soft thresholding}
\acrodef{SVD}[SVD]{singular value decomposition}

\def\Frf{{\mathbf{F}_{\mathrm{RF}}}}
\def\Frfx{{\mathbf{F}^*_{\mathrm{RF}}}}
\def\Fbb{{\mathbf{F}_{\mathrm{BB}}}}
\def\Fbbx{{\mathbf{F}^*_{\mathrm{BB}}}}
\def\Wrfx{{\mathbf{W}^*_{\mathrm{RF}}}}
\def\Wbbx{{\mathbf{W}^*_{\mathrm{BB}}}}
\def\Wrf{{\mathbf{W}_{\mathrm{RF}}}}
\def\Wbb{{\mathbf{W}_{\mathrm{BB}}}}
\newcommand{\sref}[1]{{Section}~\ref{#1}}

\newenvironment{megaalgorithm}[1][htb]{%
	\floatname{algorithm}{Transmission Method}
	\begin{algorithm}[#1]%
	}{\end{algorithm}}

\begin{document}
	%
	
	\title{Hybrid Architectures with Few-Bit ADC Receivers: Achievable Rates and Energy-Rate Tradeoffs}
	\author{\IEEEauthorblockN{Jianhua Mo, Ahmed Alkhateeb, Shadi Abu-Surra, and Robert W. Heath Jr.}\\
	\thanks{ Jianhua Mo, Ahmed Alkhateeb, and Robert W. Heath Jr. are with The University of Texas at Austin (Email: jhmo, aalkhateeb, rheath@utexas.edu).  Shadi Abu-Surra was with Samsung Research America-Dallas, Email: \{shadi.as\}@samsung.com.} \thanks{This work was done in part when the first author interned with Samsung Research America-Dallas. The authors at the University of Texas at Austin are supported in part by the National Science Foundation under Grant No. NSF-CCF-1319556 and No. NSF-CCF-1527079.}
	\thanks{The material in this paper was presented in part at the 20th International ITG Workshop on Smart Antennas in Munich, Germany, March 2016\cite{Mo_Jianhua_WSA16}.}
	}
	
	\maketitle
	\begin{abstract}
	Hybrid analog/digital architectures and receivers with low-resolution analog-to-digital converters (ADCs) are two low power solutions for wireless systems with large antenna arrays, such as millimeter wave and massive MIMO systems. Most prior work represents two extreme cases in which either a small number of RF chains with full-resolution ADCs, or low resolution ADC with a number of RF chains equal to the number of antennas is assumed. In this paper, a generalized hybrid architecture with a small number of RF chains and finite number of ADC bits is proposed. For this architecture, achievable rates with channel inversion and SVD based transmission methods are derived. Results show that the achievable rate is comparable to that obtained by full-precision ADC receivers at low and medium SNRs. A trade-off between the achievable rate and power consumption for different numbers of bits and RF chains is devised. This enables us to draw some conclusions on the number of ADC bits needed to maximize the system energy efficiency.  Numerical simulations show that coarse ADC quantization is optimal under various system configurations. This means that hybrid combining with coarse quantization achieves better energy-rate trade-off compared to both hybrid combining with full-resolutions ADCs and 1-bit ADC combining.
	
	\end{abstract}
	
	
	\newpage
	\section{Introduction}
	Massive \ac{MIMO} is a key feature of next-generation wireless systems. At low-frequencies, massive MIMO supports many users simultaneously and achieves large sum-rates with relatively simple multi-user processing \cite{Marzetta2010,Boccardi2014,Larsson2014}. At mmWave frequencies, the large antenna arrays, deployed at both the base station and mobile users, guarantee sufficient received signal power \cite{Rappaport2014,11ad,Rappaport2013a,Roh2014,Alkhateeb2014d,Heath_JSTSP16}. This allows signal transmission with ultra high data rates thanks to large bandwidths available at the mmWave frequency band \cite{11ad,Boccardi2014,Bai_Tianyang_TWC15}. Unfortunately, the high hardware cost and power consumption of mixed-signal components makes a fully-digital transmission solution, that allocates an RF chain per antenna, difficult to realize in practice \cite{Murmann_16, Hong_Wonbin_COMM14}. To overcome this challenge, new architectures that relax the requirement of associating an RF chain per antenna are being developed \cite{Alkhateeb_COMM14}. Hybrid analog/digital architectures \cite{Ayach_TWC14,Han2015}, and 1-bit ADC receivers \cite{Mo_Jianhua_TSP15} are two potential solutions. Those two solutions, though, represent two extreme cases in terms of the number of bits and RF chains. In this paper, we explore a generalization of these two architectures, where finite resolution ADCs are used with hybrid combining.
	
	
	\subsection{Related work}
	Hybrid analog/digital architectures divide the precoding/combining processing between analog and digital domains. They have been proposed for both mmWave and low-frequency massive MIMO systems \cite{Zhang2005a,Venkateswaran2010, Ayach_TWC14, Alkhateeb2013,Alkhateeb2014,Rusu_ICC15,Chen2015,Bogale2014,Sohrabi2015,Liang2014}. Hybrid architectures employ many fewer radio frequency (RF) chains than the number of antennas, relying  on RF beamforming that is normally implemented using networks of phase shifters \cite{Zhang2005a,Venkateswaran2010, Ayach_TWC14}. Hybrid precoding for diversity and  multiplexing gain was investigated in \cite{Zhang2005a}, and for interference management in \cite{Venkateswaran2010}, considering general MIMO systems. These solutions, however, did not make use of the special large MIMO characteristics in the design. For mmWave massive MIMO systems, the sparse nature of the channels was exploited to design low-complexity hybrid precoding algorithms \cite{Ayach_TWC14}, assuming perfect channel knowledge at the transmitter. Extensions to the case when only partial channel knowledge is required was considered in \cite{Alkhateeb2013,Alkhateeb2014}. Hybrid precoding algorithms that do not rely on channel sparsity were proposed in \cite{Rusu_ICC15,Chen2015}, with the aim of maximizing the system spectral efficiency. Hybrid precoding was also shown to achieve performance near that of the fully-digital solutions in low-frequency massive MIMO systems when the number of RF chains is large enough compared to the number of users \cite{Bogale2014,Sohrabi2015,Liang2014}. A common limitation of the hybrid architectures adopted in \cite{Zhang2005a,Venkateswaran_TSP10,Ayach_TWC14, Alkhateeb2013,Alkhateeb2014,Rusu_ICC15,Chen2015,Bogale2014,Sohrabi2015,Liang2014} is the assumption that the receive RF chains include high-resolution analog-to-digital converters (ADCs), which consume high power, especially at mmWave \cite{Murmann_16}. Another limitation is the extra power consumption of the analog phase shifters, which can have high impact on the energy efficiency of the hybrid combiner. Compared to the conventional fully-digital receiver, the power saved by reducing the number of RF chains in hybrid receiver may be offset by the additional power consumed by the phase shifters.

	An alternative to high resolution ADCs is to live with ultra low resolution ADCs (1-4 bits), which reduces power since ADC power grows exponentially with resolution \cite{Walden_JSAC99,Le_Bin_SPM05}. In \cite{Mezghani_ISIT07,Mezghani_WSA07, Mezghani_ISIT12, Zhang_Wenyi_TCOM12,Mo_Jianhua_TSP15, Mo_Jianhua_Asilomar14, Bai_Qing_ETT15, Orhan_ITA15, Wang_Shengchu_TWC15, Jacobsson_arxiv15}, receiver architectures where the received signal at each antenna is directly quantized by low resolution ADCs without any analog combining is considered.
	At present, the exact capacity of quantized \ac{MIMO} channel is unknown, except for special cases like the \ac{MISO} and \ac{SIMO} channels in the low or high SNR regime \cite{Mezghani_ISIT07, Singh_TCOM09, Mo_Jianhua_TSP15}. Transmitting independent QAM signals \cite{Mezghani_ISIT07} or Gaussian signals \cite{Mezghani_ISIT12, Orhan_ITA15, Bai_Qing_ETT15} from each antenna nearly achieves the capacity at low SNR, but is not optimal at high SNR. The case with CSIT was studied in our previous work \cite{Mo_Jianhua_ITA14, Mo_Jianhua_TSP15} where two methods were proposed to design the input constellation and precoder to maximize the channel capacity. It was shown that the proposed methods achieve much larger rate than QAM signaling, especially at high SNR. There is also interest in using 1-bit ADCs for the massive MIMO receiver where a large number of ADCs are needed\cite{Jacobsson_arxiv15, Wang_Shengchu_TWC15, Choi_TCOM16, Mollen_arxiv16}. The achievable rate of the multiuser uplink massive MIMO channel with 1-bit ADCs was analyzed in \cite{Jacobsson_arxiv15, Mollen_arxiv16}. Symbol detection algorithms in a similar setup were proposed in \cite{Wang_Shengchu_TWC15, Choi_TCOM16}.
	The architecture in \cite{Mezghani_ISIT07,Mezghani_WSA07, Mezghani_ISIT12, Zhang_Wenyi_TCOM12,Mo_Jianhua_TSP15, Mo_Jianhua_Asilomar14, Bai_Qing_ETT15, Orhan_ITA15, Wang_Shengchu_TWC15, Jacobsson_arxiv16, Singh_TCOM09, Mo_Jianhua_ITA14, Choi_TCOM16, Mollen_arxiv16, Studer_TCOM16}, though, assume that the number of RF chains is equal to the number of antennas,  which means that the hardware cost may be high, and no gain is made from possible beamforming processing in the RF domain.
	
	\subsection{Contribution}
	In this paper, we propose a generalized hybrid architecture with few-bit ADC receivers and draw important conclusions about its energy-rate trade-off.
	The hybrid architecture and 1-bit ADC receiver architecture studied in the past represent two extreme points in terms of the number of ADC bits and RF chains. In prior work, the hybrid architecture employs a small number of RF chains but with high resolution ADCs, while the 1-bit ADC receivers assume that the number of RF chains equals the number of antennas. The contributions of this paper are summarized as follows.
	\begin{itemize}
		\item  For the transceiver architecture with hybrid precoding/combining and low resolution ADCs, we propose two transmission methods and derive their achievable rates in closed forms. For the channel inversion based method, the inter-stream interference is canceled before quantization and hence the channel can be separated into several parallel SISO channels. For the SVD based method, the additive quantization noise model is used to derive a lower bound of the achievable rate by assuming Gaussian input distribution.
		We also derive an upper bound of the channel capacity for one-bit quantization. The bound is achieved by the proposed two transmission methods under certain conditions.
		In simulations, we show that the proposed architecture with few-bit ADCs can achieve a performance comparable to that obtained with fully-digital or hybrid architecture with infinite-bit ADC receiver in the low-to-medium SNR range, which is of a special importance for mmWave communications.

		\item We characterize the trade-off between the achievable rate and power consumption in the proposed hybrid architecture with few-bit ADC receivers. This allows us to make important conclusions about the energy efficiency of the considered hybrid architecture for different numbers of bits. This also enables us to explore the performance of the considered architecture compared with the fully-digital transceiver and the conventional hybrid architecture with full-resolution ADCs. Using numerical results and adopting a power consumption model from recent research \cite{Murmann_FTFC13,Murmann_16}, we draw insights into the optimal number of quantization bits from an energy efficiency perspective. A key finding is that coarse quantization (4-5 bits) normally achieves the maximum energy efficiency. The reason is that very low quantization (1-2 bits) suffers from a severe rate loss, while high quantization (7-8 bits) has high power consumption. Hence, both of the two regimes result in very low energy efficiency. 		
	\end{itemize}

	In conclusion, this paper draws a complete picture about the generalized hybrid architectures with few-bit ADC receivers by analyzing both their achievable spectral efficiency and their energy-rate trade-off.

	\emph{Notation} : $a$ is a scalar, $\mb{a}$ is a vector and $\mb{A}$ is a matrix. $\mr{tr}(\mb{A})$, $\mb{A}^*$ and $\|\bA\|_F$ represents the trace, conjugate transpose and Frobenius norm of a matrix $\mb{A}$, respectively. $\bI$ stands for an identity matrix. $I( \ba; \bb )$ represents the mutual information between $\ba$ and $\bb$. $\measuredangle \left(a\right)$ is the phase of the complex number $a$.
	
	\section{System Model} \label{sec:System_Model}
	\begin{figure}[t]
		\begin{centering}
			\includegraphics[width=1\columnwidth]{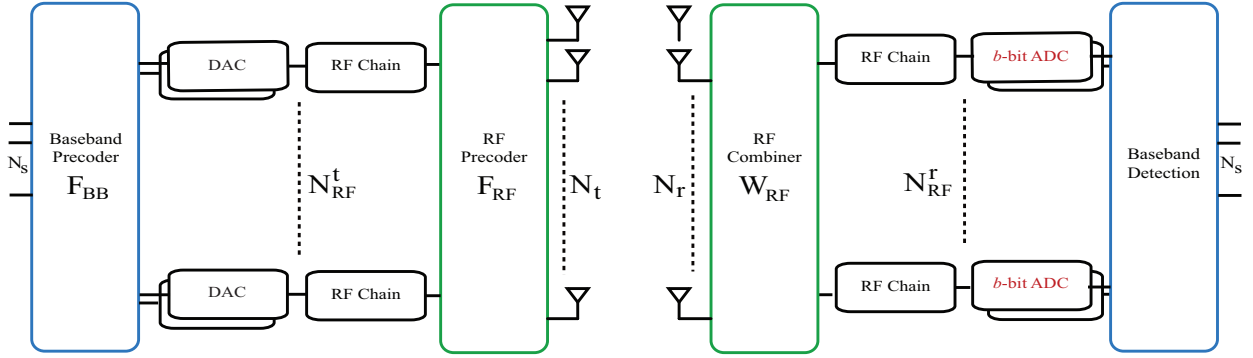}
			\vspace{-0.1cm}
			\centering
			\caption{A MIMO system with hybrid precoding and few-bit ADCs. The transmitter (receiver) has $\Nt$($\Nr$) antennas and $\Nct$($\Ncr$) RF chains. the transmitter has full-precision DACs while the receiver has only few-bit low resolution ADCs.} \label{fig:System_model}
		\end{centering}
		\vspace{-0.3cm}
	\end{figure}
	
	We propose a MIMO architecture that combines hybrid analog/digital precoding and combining with few-bit ADCs, as shown in Fig. \ref{fig:System_model}.
	The transmitter and receiver are equipped with $N_{\mathrm{t}}$ and $N_{\mathrm{r}}$ antennas, respectively. The transmitter is assumed to have $\Nct$  RF chains with full-precision digital-to-analog converters (DACs), while the receiver employs $\Ncr$ RF chains with few-bit (1-4 bits) ADCs. Further, the number of antennas and RF chains are assumed to satisfy $\left( \Nct \leq \Nt, \Ncr \leq \Nr \right)$. The transmitter and receiver communicate via $\Ns$ data streams, with $\Ns \leq \min \left(\Nct, \Ncr\right)$.
	
	Compared to the fully-digital architecture where the receiver has $\Nr$ pairs of high resolution ADCs, the proposed receiver architecture contains only $\Ncr$ pairs of few-bit ADCs, which greatly reduces both the hardware cost and power consumption. Note that the transmitter has high-resolution digital-to-analog converters (DACs) in our model. Analyzing a hybrid transceiver architecture with both low resolution ADCs and DACs is left for future work.
	
	In this paper, we denote $\Frf \in \mathbb{C}^{\Nt \times \Nct}$  as the frequency band analog precoder and $\Fbb \in \mathbb{C}^{\Nct \times \Ns}$ as the baseband digital precoder.
	Assuming a narrowband channel and perfect synchronization, the complex baseband signal prior to combining can be written as
	\begin{align}
	\mb{y} = \mb{H} \Frf \Fbb \mb{s} + \mb{n},
	\end{align}
	where $\bs$ is the digital baseband signal with the covariance $\mathbb{E}[\mb{s} \mb{s}^*]=  \frac{\Pt}{\Ns} \bI$ where $\Pt$ is the transmission power, $\mb{n} \sim \mathcal{CN}(0, \sigma_{\rmN}^2 \mb{I})$ is the white Gaussian noise with variance $\sigma_{\rmN}^2$.

    After the analog combining, quantization and digital combining, the received signal is
    \begin{align}
    \bv = \Wbbx \mathcal{Q}\left( \Wrfx \mb{H} \Frf \Fbb \mb{s} + \Wrfx \mb{n} \right),
    \end{align}
    where $\Wbb \in \mathbb{C}^{\Ncr \times \Ncr}$ is the baseband combiner, $\Wrf \in \mathbb{C}^{\Nr \times \Ncr}$ is the analog combiner, and
	$\mathcal{Q}()$ is a scalar quantization function which applies component-wise and separately to the real and imaginary parts.

	Since this paper focuses on capacity analysis and the choice of baseband combiner $\Wbb$ does not affect the channel capacity as long as $\Wbb$ is invertible, we ignore the baseband combiner in this paper. Further, we assume perfect channel knowledge at the transmitter and receiver. Developing efficient channel estimation techniques for the proposed transceiver architecture is an interesting problem for future work. These techniques may leverage the previously designed channel estimation algorithms for hybrid architectures \cite{Alkhateeb2014,Ghauch2015} and MIMO systems with low-resolution ADCs \cite{Mo_Jianhua_arxiv16b}. We denote the signal without digital combining as
	\begin{align}
	\mb{r} =  \mathcal{Q}\left( \Wrfx \mb{H} \Frf \Fbb \mb{s} + \Wrfx \mb{n} \right),
	\end{align}
	%
	where the effective noise $ \widetilde{\mb{n}} \triangleq \Wrfx \mb{n}$
	has covariance $\Wrfx \Wrf$.

	With known CSI at the transmitter, the capacity of this channel is
	\begin{align}
	C &= \max \limits_{\tiny \begin{array}{c}
			\Fbb, \Frf, \Wrf, \\
			p(\bs),  \mathcal{Q}() \end{array}} I(\bs; \br| \bH) \\
	&= \max \limits_{\tiny \begin{array}{c}
			\Fbb, \Frf, \Wrf, \\
			p(\bs),  \mathcal{Q}() \end{array}}
    \int_{\mathbf{s}} \sum_{\mathbf{r}} p(\mathbf{s}) \mathrm{Pr}(\mathbf{r}|\mathbf{s}; \mathbf{H}) \log_2 \frac{\mathrm{Pr}(\mathbf{r}|\mathbf{s}, \mathbf{H})}{\mathrm{Pr}(\mathbf{r})} \, \mathrm{d} \mathbf{s}
	\end{align}
	where $p(\bs)$ represents the probability distribution of $\bs$, $\mr{Pr}(\br|\bs; \bH)$ is the transition probability between $\bs$ and $\br$ given $\bH$, and $\Pr(\br)=\int_{\bs} p(\mathbf{s}) \mathrm{Pr}(\mathbf{r}|\mathbf{s}, \mathbf{H}) \, \mathrm{d}\mathbf{s}$. Note that the maximization is also over the quantization function $\mathcal{Q}()$, for example, thresholds of the ADCs \cite{Singh_TCOM09,Kamilov_TSP12}. If the simple uniform quantization is assumed, then the stepsize $\Delta$ is the only parameter in quantization function $\mathcal{Q}()$.

	
	
	
	
	\section{Problem Formulation}
	
	Since analog precoding (combining) is implemented by analog phase shifters, the entries of $\Frf$ ($\Wrf$) are limited to have same norm.
	The optimization problem is to maximize the mutual information between $\bs$ and $\br$ as follows.
    \small
	\begin{align}
	\mr{P1:}  \max \limits_{\tiny \begin{array}{c}
			\Fbb, \Frf, \Wrf, \\
			p(\bs),  \mathcal{Q}() \end{array}} & I(\mb{s}; \mb{r}| \bH) \\
	\mr{s.t.} \quad & \Big| [\Frf]_{mn} \Big| = \frac{1}{\sqrt{\Nt}}, \quad \forall m, n, \label{eq:Frf_unit_norm}\\
	      & \Big| [\Wrf]_{mn} \Big| = \frac{1}{\sqrt{\Nr}},  \quad \forall m, n, \label{eq:Wrf_unit_norm} \\
	      & \|\Frf \Fbb\|_F^2 = \Ns,  \label{eq:power_constraint}
	\end{align}
    \normalsize
	where \eqref{eq:power_constraint} is due to the transmission power constraint, i.e., $\mathbb{E} \left[ \|\Frf \Fbb \bs \|^2 \right] = \Pt$.
	
	It is very non-trivial to solve the problem P1. First, it is hard to optimize the mutual information over so many parameters simultaneously. Second, the quantization function $\mathcal{Q}()$ is nonlinear and also related to the input distribution $p(\bs)$ and renders it difficult to analyze the mutual information.
	Third, the equality constraints in \eqref{eq:Frf_unit_norm} and \eqref{eq:Wrf_unit_norm} is non-convex and hard to deal with.
	

	%
	Throughout the paper, the analog precoder $\Frf$ is assumed to satisfy $\Frfx \Frf = \bI$. Under this assumption, the coupled power constraint  \eqref{eq:power_constraint} involving the digital and analog precoding become a simple constraint on the digital precoder $\Fbb$. The similar assumption also appeared in \cite{Rusu_ICC15} where the digital and analog precoders are designed separately.
In addition, we also assume that $\Wrfx \Wrf = \bI $ and therefore the effective noise is still white Gaussian noise. This assumption simplifies the computation of mutual information and a similar idea appeared in \cite{Molisch_CL04}. Further, $\Frfx \Frf$ and $\Wrfx \Wrf$ are approximately to be identity matrices when $\Nt$ and $\Nr$ is large.
	To sum up, we assume both the analog precoder and combininer are semi-unitary matrices.
	
	Consequently, the optimization problem P1 is reformulated as
    \small
	\begin{align}
	\mr{P2:}  \max \limits_{\tiny \begin{array}{c}
			\Fbb, \Frf, \Wrf, \\
			p(\bs),  \mathcal{Q}() \end{array}} & I(\mb{s}; \mb{r}|\bH) \\
	\mr{s.t.} \quad & \Big| [\Frf]_{mn} \Big| = \frac{1}{\sqrt{\Nt}}, \quad \forall m, n, \\
	     & \Big| [\Wrf]_{mn} \Big| = \frac{1}{\sqrt{\Nr}}, \quad \forall m, n, \\
	     & \Frfx \Frf =  \bI, \quad \Wrfx \Wrf =  \bI, \label{eq:semi_unitary} \\
	     & \left\| \Fbb \right\|_F^2 \leq \Ns.
	\end{align}
    \normalsize
	
In this paper, we develop two transmission strategies, including the precoding techniques, the distribution of signal $\bs$, and the quantization design. We will investigate their achievable rates and show that their performance are close to optimum in certain cases.
	
	\section{Upper Bound of the Achievable Rate}
	
	In this section, we provide upper bounds of the achievable rate for one-bit quantization. The upper bounds are used as benchmarks for evaluating our proposed transmission methods.
	For multi-bit quantization, the upper bounds are unknown and left for future work.
	
	
	
	
	\begin{proposition} \label{prop:one_bit_ub}
		An upper bound on the achievable rate with hybrid precoding and one-bit quantization is
		\begin{align} \label{eq:ub_tight}
		R^{\mr{1bit}, \mr{ub}} = 2 \Ncr \left(1 - \Hb \left( Q \left(\sqrt{\frac{\rho \nu^2_{1}}{ \Ncr }}\right) \right) \right),
		\end{align}
		where $\Hb(x) = - x \log_2 x - (1-x) \log_2 (1-x)$ is the binary entropy function, $Q(\cdot)$ is tail probability of the standard normal distribution, $\rho \triangleq \frac{\Pt}{\sigma_{\rmN}^2}$ is the SNR and $\nu_1$ is the maximum singular value of the effective channel matrix $\bG \triangleq \Wrfx \bH \Frf$.
	\end{proposition}

	\begin{proof} Please see the appendix.
	\end{proof}
    The upper bound is achieved when the effective channel $\bG$ is full rank and has $\Ncr$ identical singular values, or equivalently, $\bG \bG^*= \nu_1^2 \bI$.
    	
	At low SNR, this upper bound is approximated as
	\begin{align} \label{eq:R_1bit_ub_low_SNR}
	R^{\mr{1bit}, \mr{ub}} = \frac{2}{\pi} \frac{\rho \nu_1^2}{\ln 2} + o(\rho),
	\end{align}
	following the facts $Q(t) = \frac{1}{2} - \frac{1}{\sqrt{2 \pi}}t + o(t^2)$ and $\Hb(\frac{1}{2}+t) = 1 - \frac{2}{\ln 2} t^2 + o(t^2)$. Therefore the bound increases linearly with the power at low SNR. But at high SNR, the upper bound converges to $2 \Ncr$ bps/Hz, which is due to the finite number of quantization output bits.

	Note that the upper bound given in \eqref{eq:ub_tight} is related to the choice of analog precoding $\Wrf$ and $\Frf$.
	Next, we give another bound, which is looser but independent of the analog precoding.
	\begin{corollary}
		An upper bound of the achievable rate with hybrid precoding and one-bit quantization is
		\begin{align} \label{eq:ub_loose}
		\widetilde{R}^{\mr{1bit}, \mr{ub}} = 2 \Ncr \left(1 - \Hb \left( Q \left(\sqrt{\frac{\rho {\sigma}^2_{1}}{ \Ncr }}\right) \right) \right),
		\end{align}
		where $\sigma_1$ is the maximum singular value of $\bH$.
	\end{corollary}
	\begin{IEEEproof}
		Under the constraint $\Wrfx \Wrf =  \bI $ and $\Frfx \Frf =  \bI$, it is proved that $\nu_1^2 \leq {\sigma}_1^2$ in \cite[Theorem 2.2]{Rao_Math79}.
		Therefore, we have
		\begin{align}
		R^{\mr{1bit}, \mr{ub}}  \leq  2 \Ncr \left(1 - \Hb \left( Q \left(\sqrt{\frac{\rho {\sigma}^2_{1}}{\Ncr}}\right) \right) \right) \triangleq \widetilde{R}^{\mr{1bit}, \mr{ub}}.
		\end{align}
		This completes the proof of Corollary 1.
	\end{IEEEproof}
	\begin{remark}
		For a channel with infinite-bit ADCs, a simple upper bound of the capacity is $\Ncr \log_2 \left(1 + \frac{\rho \nu_1^2}{\Ncr}\right)$, which is achieved when the effective channel $\bG$ has same singular values.
		Compared to \eqref{eq:ub_tight} and \eqref{eq:ub_loose}, the bound for infinite-bit ADCs increases to infinity as the power increases to infinity.
	\end{remark}
	

	\section{Achievable Rate with Channel Inversion Based Transmission}
	In this section, we propose channel inversion based transmission. In this transmission method, there is no interference among data streams at the receiver and each stream is quantized separately.
Therefore, the exact achievable rate of this method can be found in closed-form.
	
	\subsection{Channel Inversion Based Precoding Algorithm}
	For digital precoding design, we propose to use channel inversion precoding assuming that $\Nct \geq \Ncr = \Ns$.
	The digital precoder is
	\begin{align} \label{eq:Fbb_CI}
	\Fbb = \sqrt{\frac{\Ns}{\beta}}
	\mb{G}^*  \left( \mb{G} \mb{G}^* \right)^{-1}
	\end{align}
	where
	\begin{align}
	\beta = \mr{tr}\left\{ {\bG}^* \left( {\bG} {\bG}^* \right)^{-2} {\bG} \Frfx \Frf \right\}
	\end{align}
	such that the power constraint \eqref{eq:power_constraint} is satisfied.
	As it is assumed that $\Frfx \Frf = \bI$, $\beta$ is simplified to be
	\begin{align}
	\beta =  \mr{tr}\left\{\left( \mb{G} \mb{G}^* \right)^{-1}   \right\}.
	\end{align}
	
	Since there is no interference among streams because of channel inversion precoding, each stream of data can be detected separately.
	The received signal is
	\begin{align}
	\mb{r} &= \mathcal{Q} \left( \Wrfx \mb{H} \Frf \Fbb \mb{s} +  \Wrfx \mb{n} \right) \\
	&= \mathcal{Q} \left( \sqrt{\frac{\Ns}{\mr{tr}\left\{\left( \mb{G} \mb{G}^* \right)^{-1}   \right\}}} \mb{s} +  \Wrfx \mb{n} \right).
	\end{align}
	
	The channel is converted to $2 \Ns$ parallel sub channels, each of which is a quantized real-valued \ac{SISO} channel. The SNR of each sub-channel is given by
	\begin{align} \label{eq:SNR}
	\mr{SNR_{CI}}= \frac{\rho}{\mr{tr}\left\{ \left( \mb{G} \mb{G}^* \right)^{-1} \right\} },
	\end{align}
	where $\rho \triangleq \frac{\Pt}{\sigma_{\rmN}^2}$.
	
	Maximizing the SNR is equivalent to maximizing the following term
	\begin{align}
	\eta \left( \bG \right) & \triangleq \left(\mr{tr}\left\{ \left( \mb{G} \mb{G}^* \right)^{-1} \right\} \right)^{-1} \\
	&= \left(\frac{1}{\nu_1^2} + \frac{1}{\nu_2^2} + \cdots + \frac{1}{\nu_{\Ns}^2} \right)^{-1},
	\end{align}
	where $\nu_1$, $\nu_2$, $\cdots$, $\nu_{\Ns}$ are the singular values of the effective channel $\mb{G}$ in descending order.
	Therefore, $\Wrf$ and $\Frf$ should be chosen to maximize the harmonic mean of the squared singular values of $\mb{G}$, or equivalently the harmonic mean of the eigenvalues of $ \mb{G} \mb{G}^*$.
	
	To maximize $\eta \left(\bG \right)$, the optimal choice of $\Wrf$ and $\Frf$ are the singular vectors associated with the largest $\Ns$ singular values of $\bH$  \cite{Palomar_TSP03}. Although such choice satisfies the semi-unitary constraints \eqref{eq:semi_unitary}, the norm constraints \eqref{eq:Frf_unit_norm}-\eqref{eq:Wrf_unit_norm} are violated.
	
	In this paper, we use alternating projection algorithm \cite{Tropp_IT05} to find an approximate solution satisfying both the constant-norm and semi-unitary constraints. The algorithm is summarized in Algorithm \ref{alg:AltProj}. In steps 2a)-2b), the semi-unitary matrices $\widehat{\bW}_{\mr{RF}}$ and $\widehat{\bF}_{\mr{RF}}$ are projected to the sets of matrices satisfying the norm constraints \eqref{eq:Frf_unit_norm}-\eqref{eq:Wrf_unit_norm}, resulting in the solutions
	$\widetilde{\bW}_{\mr{RF}}$ and $\widetilde{\bF}_{\mr{RF}}$ respectively. Each element of $\widetilde{\bW}_{\mr{RF}}$ $\left(\widetilde{\bF}_{\mr{RF}}\right)$ has the same phase of the corresponding element in $\widehat{\bW}_{\mr{RF}}$ $\left(\widehat{\bF}_{\mr{RF}} \right)$ but normalized amplitude.
	In steps 2c)-2d), $\widetilde{\bW}_{\mr{RF}}$ and $\widetilde{\bF}_{\mr{RF}}$ are projected back to the sets of semi-unitary matrices. The projection process continues until $\frac{\left\|\widehat{\mathbf{F}}^{(k)}_{\mathrm{RF}} - \widetilde{\mathbf{F}}^{(k)}_{\mathrm{RF}} \right\|_F}{\sqrt{\Nct}} \left(\frac{\left\|\widehat{\mathbf{W}}^{(k)}_{\mathrm{RF}} - \widetilde{\mathbf{W}}^{(k)}_{\mathrm{RF}} \right\|_F}{\sqrt{\Ncr}}\right)$ is smaller than a specified threshold $\epsilon$. The convergence of the alternating projection algorithm is discussed in details in \cite{Tropp_IT05}. In \figref{fig:AltProj_Convergence}, we show a typical convergence behaviour when the transmitter is assumed to have  64 antennas and 8 RF chains, while the receiver employs 8 antennas and 4 RF chains. Therefore, $\Frf \in \mathbb{C}^{64\times 8}$ and $\Wrf \in \mathbb{C}^{8 \times 4}$. Note that $\left\|\widetilde{\bF}_{\mr{RF}}\right\|_F^2 = \left\|\widehat{\bF}_{\mr{RF}}\right\|_F^2 = \Nct$ and $\left\|\widetilde{\bW}_{\mr{RF}}\right\|_F^2 = \left\|\widehat{\bW}_{\mr{RF}}\right\|_F^2 = \Ncr$. So the distance is normalized by $\sqrt{\Nct}$ and $\sqrt{\Ncr}$, respectively. It is seen the algorithm converges very fast, within less than 100 iterations, to a normalized distance of less than $10^{-5}$.

Another choice for designing analog precoder is to assume that $\Frf$ and $\Wrf$ consist of columns from the DFT matrices\cite{Mo_Jianhua_WSA16}. This is inspired by the virtual channel representation \cite{Sayeed_TSP02}. Note that the DFT matrix has constant-norm entries and orthogonal columns, therefore the norm and semi-unitary constraints of analog precoder are both satisfied. However, searching the best combination of columns has higher complexity than the alternating projection method when the number of antennas is large.

	\begin{algorithm}[t]
		\caption{Alternating projection algorithm for analog precoding design}
		\label{alg:AltProj}
		\begin{enumerate}
			\item Initialize $\widehat{\bW}^{(0)}_{\mr{RF}} = \bU$ and $\widehat{\bF}^{(0)}_{\mr{RF}} = \bV$ where $\bH = \bU \mathbf{\Sigma} \bV^*$ is the singular value decomposition of $\bH$. Set $k=1$ and $\epsilon=10^{-5}$.
			\item Alternating projection method
				\begin{enumerate}
					\item $ \left[\widetilde{\bW}^{(k)}_{\mr{RF}} \right]_{mn} = \frac{1}{\sqrt{\Nr}} \exp \left(\j \measuredangle \left( \left[\widehat{\bW}^{(k-1)}_{\mr{RF}}\right]_{mn} \right) \right)$,
					$\forall m, n,$
					\item $ \left[\widetilde{\bF}^{(k)}_{\mr{RF}} \right]_{mn} = \frac{1}{\sqrt{\Nt}} \exp \left(\j \measuredangle \left( \left[\widehat{\bF}^{(k-1)}_{\mr{RF}}\right]_{mn} \right) \right)$,
					$\forall m, n,$
					\item $\widehat{\bW}^{(k)}_{\mr{RF}} = \widetilde{\bW}^{(k)}_{\mr{RF}} \left( \left(\widetilde{\bW}^{(k)}_{\mr{RF}}\right)^{*} \widetilde{\bW}^{(k)}_{\mr{RF}} \right)^{-\frac{1}{2}}$,
					\item $\widehat{\bF}^{(k)}_{\mr{RF}} = \widetilde{\bF}^{(k)}_{\mr{RF}} \left( \left(\widetilde{\bF}^{(k)}_{\mr{RF}}\right)^{*} \widetilde{\bF}^{(k)}_{\mr{RF}} \right)^{-\frac{1}{2}}$,
                    \item If the normalized distance $\frac{\left\|\widehat{\mathbf{W}}^{(k)}_{\mathrm{RF}} - \widetilde{\mathbf{W}}^{(k)}_{\mathrm{RF}} \right\|_F}{\sqrt{\Ncr}} < \epsilon$ and $\frac{\left\|\widehat{\mathbf{F}}^{(k)}_{\mathrm{RF}} - \widetilde{\mathbf{F}}^{(k)}_{\mathrm{RF}} \right\|_F}{\sqrt{\Nct}} < \epsilon$, return $\widehat{\bW}^{(k)}_{\mr{RF}}$ and $\widehat{\bF}^{(k)}_{\mr{RF}}$; \\
                        else, $k=k+1$ and go back to step (a).
				\end{enumerate}
		\end{enumerate}
	\end{algorithm}

	\begin{figure}[t]
		\begin{centering}
			\includegraphics[width=.7\columnwidth]{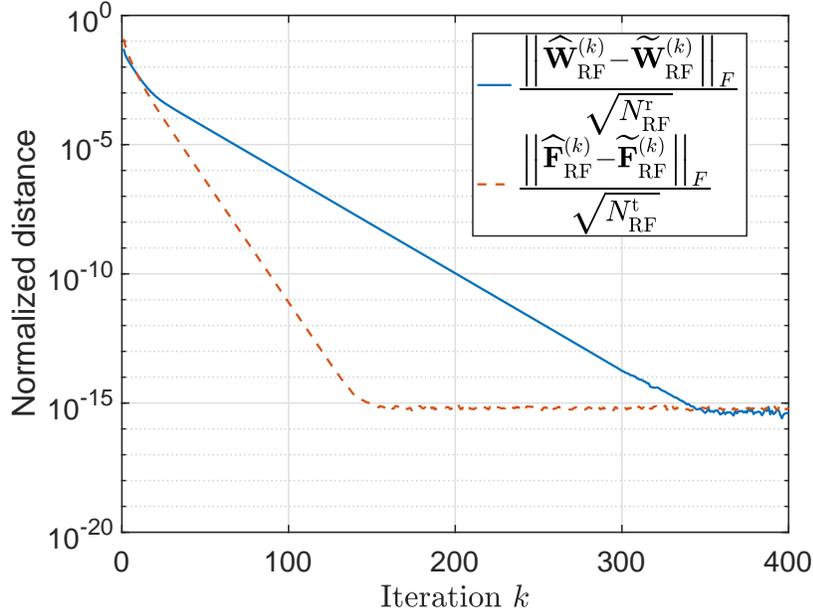}
			\vspace{-0.1cm}
			\centering
			\caption{This figure shows a typical convergence behaviour of the alternating projection method. The normalized error decreases exponentially with the iteration number $k$. In the figure, $\Frf \in \mathbb{C}^{64\times 8}$ and $\Wrf \in \mathbb{C}^{8 \times 4}$. Note that the floor of the normalized distance around $10^{-15}$ is due to the precision limitation of our computer. By default, MATLAB uses 16 digits of precision.}
			\label{fig:AltProj_Convergence}
		\end{centering}
		\vspace{-0.3cm}
	\end{figure}

	\subsection{Rate Analysis with One-Bit Quantization}
	In this subsection, we focus on the special case of one-bit quantization and derive the gap between the achievable rate and the upper bound given in Proposition \ref{prop:one_bit_ub}.
	
	If one-bit ADCs are used at the receiver, the capacity of each real-valued \ac{SISO} sub-channel is achieved by binary antipodal signaling and is given by \cite{Singh_TCOM09}
	\begin{align}
	1 - \Hb \left( Q \left(\sqrt{ \mr{SNR_{CI}} }\right) \right).
	\end{align}
	The total sum rate therefore is
	\begin{align} \label{eq:R_ci_1bit}
	R_{\mr{CI}}^{\mr{1bit}} = 2 \Ns \left(1 - \Hb \left( Q \left(\sqrt{ \mr{SNR_{CI}} }\right) \right) \right).
	\end{align}
	Next, noting that
	\begin{align}
	\left(\frac{1}{\nu_1^2} + \frac{1}{\nu_2^2} + \cdots + \frac{1}{\nu_{\Ns}^2} \right)^{-1}
	\geq  \frac{\nu_{\Ns}^2}{\Ns},
	\end{align}
	a lower bound of the SNR of the proposed precoding design is
	\begin{align} \label{eq:SNR_lb}
	\mr{SNR_{CI}} \geq \frac{\rho \nu^2_{\Ns}}{ \Ns } \triangleq \mr{SNR_{lb}}.
	\end{align}
	
	%
	
	Based on the SNR lower bound in \eqref{eq:SNR_lb}, a lower bound of the achievable rate is
	\begin{align}
	R_{\mr{CI}}^{\mr{1bit}, \mr{lb} } & =  2 \Ns \left(1 - \Hb \left( Q \left(\sqrt{\frac{\rho \nu^2_{\Ns}}{ \Ns }}\right) \right) \right) \label{eq:R_ci_1bit_lb} \\
	&= 2 \Ns \left(1 - \Hb \left( Q \left(\sqrt{\frac{\rho \nu^2_{1}}{ \Ns }  \frac{\nu^2_{\Ns}}{\nu^2_{1}}}\right) \right) \right).  \label{eq:R_ci_1bit_compare}
	\end{align}
	
	Comparing \eqref{eq:R_ci_1bit_compare} and \eqref{eq:ub_tight}, we find that the power gap between $R^{\mr{1bit}, \mr{ub}}$ and $R_{\mr{CI}}^{\mr{1bit}, \mr{lb}}$ is $10 \log_{10} \frac{\nu_1^2}{\nu_{\Ns}^2}$ dB. Therefore, we conclude that compared to optimal digital precoding (which is unknown), the power loss of the channel inversion precoding is at most $10 \log_{10} \frac{\nu_1^2}{\nu_{\Ns}^2}$ dB. It also implies that for a well-conditioned effective channel, the power loss is small.
	
	Furthermore, if there is only one RF chain at the receiver, i.e., $\Ncr=\Ns=1$, the achievable rate in  \eqref{eq:R_ci_1bit} and the upper bound in \eqref{eq:ub_tight} are exactly same and it implies that the channel capacity is achieved by the proposed transmission method.

    At last, if there are more RF chains at the receiver than that at the transmitter, i.e., $\Ncr > \Nct$, then only $\Nct$ (or less than $\Nct$) out of $\Ncr$ receive RF chains are used, otherwise $\bG \bG^*$ in \eqref{eq:Fbb_CI} does not have the inverse. The power consumption also decreases by turning off some receive RF chains.
	
	\subsection{Rate Analysis with Few-Bit Quantization}
	For a real-valued SISO channel with $b$-bit quantizer, a benchmark design was proposed in \cite{Singh_TCOM09} where input signal are equiprobable, equispaced $2^b$-PAM (pulse amplitude modulated), and quantizer thresholds are chosen to be the mid-points of the input mass point locations. Although this combination of input and quantization is suboptimal, it is shown in \cite{Singh_TCOM09} to be close to the optimum which is obtained by high-complexity iterative algorithm. In addition, this combination is actually optimal for one-bit quantization.
	Therefore, in the proposed method, we assume the simple $2^{2b}$-QAM signaling at the transmitter and uniform quantization at the receiver.

    The channel inversion based transmission, including the analog and digital precoding design, signaling and quantization, is summarized in Transmission Method \ref{alg:CI}.

	\begin{megaalgorithm}[h]                      
		\caption{Channel Inversion Based Transmission Method}          
		\label{alg:CI}                           
		\begin{enumerate}
			\item  \textbf{Analog precoding design}: Find the approximate solution $\overline{\bW}_{\mr{RF}}$ and $\overline{\bF}_{\mr{RF}}$ by the alternating projection shown in Algorithm \ref{alg:AltProj}.
			
			\item \textbf{Digital precoding design}:
			\begin{enumerate}
				\item Compute the effective channel $\overline{\bG} \triangleq \overline{\bW}^*_{\mr{RF}} \bH \overline{\bF}_{\mr{RF}}$.
				\item Set the digital precoder $\Fbb$ as
				\begin{align*}
				\Fbb = \sqrt{\frac{\Ns}{ \mr{tr}\left\{  \left( \overline{\bG} \, \overline{\bG}^* \right)^{-1}  \right\}}}
				\overline{\bG}^*  \left( \overline{\bG} \, \overline{\bG}^* \right)^{-1}.
				\end{align*}
			\end{enumerate}
			\item \textbf{Signaling}: $\bs$ is chosen to be $2^{2b}$-QAM symbols.
			\item \textbf{Quantization}: Uniform quantization.
    	\end{enumerate}
	\end{megaalgorithm}
	
	We next show an example of two-bit quantization. The derivation with multi-bit quantization is similar.
	The set of input signals is $\mathcal{S} = \left\{-\frac{3 \Delta}{2}, -\frac{ \Delta}{2}, \frac{ \Delta}{2}, \frac{3\Delta}{2} \right\}$ where $\Delta$ is the stepsize.
	The transition probability matrix is
	\begin{align}
	\Pr(r|s) = &
	\left[
	\begin{array}{cccc}
	\Phi(\frac{\Delta}{2 \xi}) & \Phi(\frac{3\Delta}{2 \xi}) - \Phi(\frac{\Delta}{2 \xi}) & \Phi(\frac{5\Delta}{2 \xi}) - \Phi(\frac{3\Delta}{2 \xi}) & 1 - \Phi(\frac{5\Delta}{2 \xi}) \\
	\Phi(\frac{-\Delta}{2 \xi}) & \Phi(\frac{\Delta}{2 \xi}) - \Phi(\frac{-\Delta}{2 \xi}) & \Phi(\frac{3\Delta}{2 \xi}) - \Phi(\frac{\Delta}{2 \xi}) & 1 - \Phi(\frac{3\Delta}{2 \xi}) \\
	\Phi(\frac{-3\Delta}{2 \xi}) & \Phi(\frac{-\Delta}{2 \xi}) - \Phi(\frac{-3\Delta}{2 \xi}) & \Phi(\frac{\Delta}{2 \xi}) - \Phi(\frac{-\Delta}{2 \xi}) & 1 - \Phi(\frac{\Delta}{2 \xi}) \\
	\Phi(\frac{-5\Delta}{2 \xi}) & \Phi(\frac{-3\Delta}{2 \xi}) - \Phi(\frac{-5\Delta}{2 \xi}) & \Phi(\frac{-\Delta}{2 \xi}) - \Phi(\frac{-3\Delta}{2 \xi}) & 1 - \Phi(\frac{-\Delta}{2 \xi}) \\
	\end{array}
	\right]
	\begin{array}{l}
	s = - \frac{3 \Delta}{2}\\
	s = - \frac{\Delta}{2}\\
	s = \frac{\Delta}{2} \\
	s = \frac{3 \Delta}{2}
	\end{array} \\
	& \begin{array}{c}
	r = - \frac{3 \Delta}{2} \quad \quad \quad \;
	r = - \frac{\Delta}{2} \quad \quad \quad \;
	r = \frac{\Delta}{2}  \quad \quad \quad \quad \quad
	r = \frac{3 \Delta}{2} \quad \quad \quad
	\end{array}  \nonumber
	\end{align}
	where $\xi^2$ denotes the noise variance and $\Phi(\cdot)$ is the cumulative distribution function of the standard normal distribution. The transition probability matrix of higher resolution ADCs could be obtained similarly.

	The SNR of each sub-channel must be equal to the value given in \eqref{eq:SNR}. Therefore,
	\begin{align}
	\frac{1}{ 2^{b+1} } \left( \Delta^2 + (3\Delta)^2 + \cdots +\left( \left( 2^b-1 \right) \Delta \right)^2 \right)/ \xi^2 &= \mr{SNR_{CI}} \\
	\frac{1}{2^{b+1}} \frac{1}{3} 2^{b-1} (2^{2b} - 1) \frac{\Delta^2}{\xi^2} &= \mr{SNR_{CI}} \label{eq:summation} \\
	\frac{\Delta}{\xi} &= \sqrt{ \frac{12 \, \mr{SNR_{CI}} }{ 2^{2b} - 1}}
	\end{align}
	where \eqref{eq:summation} is from the fact that $1^2 + 3^2 + \cdots + (2n-1)^2 = \frac{1}{3} n (4 n^2-1)$.
	
	The achievable rate can therefore be computed as
	\begin{align}
	R_{\mr{CI}}^{ \mr{b \, bit} } & = 2 \Ns \sum_{s} \sum_{r} \Pr(s) \Pr(r|s) \log \frac{\Pr(r|s)} { \Pr(r)} \\
	&= 2 \Ns \sum_{s} \sum_{r} \Pr(s) \Pr(r|s) \log \frac{\Pr(r|s)} { \sum_{s'} \Pr(s') \Pr(r|s')} \\
	&\stackrel{(a)}{\leq} 2 \Ns \sum_{s} \sum_{r} \frac{1}{2^b} \Pr(r|s) \log \frac{\Pr(r|s)} { \frac{1}{2^b} \sum_{s'} \Pr(r|s')} \\
	&= 2 \Ns \left(b + \frac{1}{2^b} \sum_{s} \sum_{r}  \Pr(r|s) \log \frac{\Pr(r|s)} { \sum_{s'} \Pr(r|s')} \right), \label{eq:R_ci_bbit}
	\end{align}
	where (a) follow from that the equal transmission probability $\Pr(s)=\frac{1}{2^b}$.
	The achievable rate in \eqref{eq:R_ci_bbit} is complicated and does not provide us with enough intuition. A simple lower bound of \eqref{eq:R_ci_bbit} can be found by Fano's inequality \cite[Section 2.10]{Cover_Book12}. The conditional entropy is upper bounded by
	\begin{align}
	\Hb (s|r) \leq \Hb (P_e)+ P_e \log (|\mathcal{S}|-1) ,
	\end{align}
	where $P_e$ is the error probability.
	Hence, the mutual information between $s$ and $r$ is
	\begin{align}
	I(s;r) &= \Hb (s) - \Hb (s|r) \\
	&\geq b - \Hb (P_e) - P_e \log \left(2^b - 1 \right).
	\end{align}
	A lower bound of the sum rate of $2\Ns$ sub-channels therefore is
	\begin{align} \label{eq:R_ci_lb_bbit}
	 R_{\mr{CI}}^{\mr{b\, bit}, \mr{lb}} =  2 \Ns \left( b - \Hb (P_e) - P_e \log \left(2^b - 1 \right) \right),
	\end{align}
    where the error probability $P_e$ for $2^b$-PAM signaling is \cite{Proakis_Book08}
	\begin{align}
	P_e &= 2 \left(1 - \frac{1}{2^b} \right) Q \left( \frac{\Delta}{2 \xi} \right) \\
	&= 2 \left(1 - \frac{1}{2^b} \right) Q \left( \sqrt{ \frac{3 \, \mr{SNR}_{\mr{CI}} }{ 2^{2b} - 1 } } \right).
	\end{align}
	
	From \eqref{eq:R_ci_lb_bbit}, we find that the as $\mr{SNR}_{\mr{CI}}$ increases, $P_e$ decreases to zero and $R_{\mr{CI}}^{\mr{b\, bit}, \mr{lb}}$ converges to $2 \Ns b$ bps/Hz.
	%
	In addition, note that for the one-bit case, $P_e = Q \left(\sqrt{ \mr{SNR_{CI}} }\right)$ and therefore \eqref{eq:R_ci_lb_bbit} degrades to \eqref{eq:R_ci_1bit}.

	\section{Achievable Rate With Singular Value Decomposition Based Transmission}
	The channel inversion precoding generally works well at high SNR and has poor performance at low SNR. In the second transmission method, we use the
	\ac{SVD} digital precoding. Since the interference between each steams can not be completely eliminated before quantization as in the channel inversion case, the exact achievable rate is unknown. We therefore choose to apply the additive quantization noise model (AQNM) \cite{Fletcher_JSTSP07, Mezghani_ISIT12, Bai_Qing_ETT15}, which is accurate enough at low SNR, to find a lower bound of the achievable rate.
	
	Applying the additive quantization noise model, the equivalent channel is
	\begin{align}
	\mb{r} =  \left( 1 - \eta_b \right) \left( \bG \Fbb \mb{s} + \Wrfx \mb{n} \right) + \bn_{\mr{Q}},
	\end{align}
	where $\eta_b$ is the distortion factor for $b$-bit ADC defined as $\eta_b = \frac{\mathbb{E} \left[\left(\bQ(y)-y\right)^2 \right]}{\mathbb{E} \left[|y|^2 \right]}$ and $\bn_{\mr{Q}}$ is the quantization noise with the variance $\eta_b (1-\eta_b) \mathrm{diag}\left\{ \frac{\Pt}{\Ns} \bG \Fbb \Fbbx \bG^{*} + \sigma_{\rmN}^2 \bI \right\}$.
	
	Assuming $\bs$ is Gaussian distributed and the quantization noise is the worst case of Gaussian distributed, a lower bound of the achievable rate is
	\begin{align} \label{eq:R_AQNM_lb}
	R_{\mr{AQNM}} = \log_2 \left|\bI +  (1-\eta_b) \frac{\rho}{\Ns} \Fbbx \mb{G}^*  \left( \bI + \eta_b \, \mathrm{diag} \left\{ \frac{\rho}{\Ns}\mb{G} \Fbb \Fbbx \mb{G}^{*} \right\}  \right)^{-1} \mb{G} \Fbb \right|.
	\end{align}
	At low SNR, the achievable rate is approximated to be
	\begin{align}
	R_{\mr{AQNM}} & =  \log_2 \left|\bI +  \left( 1-\eta_b \right) \frac{\rho}{\Ns} \Fbbx \mb{G}^* \mb{G} \Fbb  \right| + o(\rho)\\
	& =   \frac{ \left( 1-\eta_b \right) \rho}{\Ns \ln 2} \mr{tr} \left\{\Fbbx \mb{G}^* \mb{G} \Fbb  \right\} + o(\rho).
	\end{align}
	
	To maximize the term $\frac{1}{\Ns} \mr{tr} \left\{ \mb{G} \Fbb \Fbbx \mb{G}^*  \right\}$ under the constraint $\|\Fbb\|_F^2 \leq \Ns$, the optimal choice of $\Fbb$ is the eigenmode beamforming, i.e.,
	\begin{align}
	\Ns=1, \quad \text{and} \quad \Fbb = \mb{v}_1,
	\end{align}
	where $\bv_1$ is the right singular vector corresponding to the largest singular value of $\bG$. The resulting rate is
	\begin{align}
		R_{\mr{AQNM}} = \frac{ (1-\eta_b) \rho \nu_1^2}{\ln 2}  + o(\rho).
	\end{align}
	For one-bit quantization, $\eta_1 = 1 - \frac{2}{\pi} $ and the rate is
	\begin{align}
		R_{\mr{AQNM}} = \frac{2}{\pi} \frac{\rho \nu_1^2}{\ln 2} + o(\rho).
	\end{align}
	Notice that the upper bound of one-bit quantized channel at low SNR given in \eqref{eq:R_1bit_ub_low_SNR} is achieved by eigenmode beamforming.
	
	For higher SNR, the optimal $\Fbb$ maximizing the rate $R_{\mr{AQNM}}$ in \eqref{eq:R_AQNM_lb} is unknown. We therefore use the conventional \ac{SVD} precoding and waterfilling power allocation as done in \cite{Mezghani_WSA07, Orhan_ITA15}. The baseband digital precoder is
	\begin{align}
		\Fbb = \bV \, \mr{diag} \left\{ \sqrt{\mb{p}} \right\},
	\end{align}
	where $\bV$ is obtained from the singular value decomposition of the matrix $\bG$, i.e., $\bG = \bU \mb{\Sigma} \bV^*$ and $\bp$ denotes the power allocation factor obtained from the waterfilling method.
	
	
	For the analog precoding and combining, the optimal $\Frf$ and $\Wrf$ the rate in \eqref{eq:R_AQNM_lb} are unknown. We adopt the same alternating projection method in Algorithm \ref{alg:AltProj} to find a suboptimal solution. As a result, the analog precoding is designed by alternating projection method with initial values obtain by SVD of the channel and the the digital precoding is got by SVD of the baseband channel.
    The proposed SVD based design is summarized in Transmission Method \ref{alg:SVD}.
	
	
	\begin{megaalgorithm}[h]                      
		\caption{\ac{SVD} Based Transmission Method}          
		\label{alg:SVD}                           
		\begin{enumerate}
			\item \textbf{Analog precoding design}: Find the approximate solution $\overline{\bW}_{\mr{RF}}$ and $\overline{\bF}_{\mr{RF}}$ by the alternating projection method shown in Algorithm \ref{alg:AltProj}.
			
%
			\item \textbf{Digital precoding design}:
			\begin{enumerate}
				\item Compute the effective channel $\overline{\bG} \triangleq \overline{\bW}^*_{\mr{RF}} \bH \overline{\bF}_{\mr{RF}}$.
				\item Set the digital precoder $\Fbb$ by \ac{SVD} of $\overline{\bG}$ and waterfilling method.
			\end{enumerate}
			\item \textbf{Signaling}: $\bs$ is chosen to follow Gaussian signaling.
			\item \textbf{Quantization}: The thresholds of ADC are determined by Max-Lloyd  algorithm \cite{Max_IRE60, Lloyd_IT82} which minimizes the MSE of Gaussian distributed input.
		\end{enumerate}            
	\end{megaalgorithm}
	
	Last, we show why the AQNM model is not accurate enough at high SNR to model the quantization channel.
	At high SNR, $R_{\mr{AQNM}}$ converges as follows.
	\begin{align}
	R_{\mr{AQNM}}
	&\approx  \log_2 \left|\bI +  \frac{1-\eta_b}{\eta_b} \Fbbx \mb{G}^* \left( \mathrm{diag} \left\{ \mb{G} \Fbb \Fbbx \mb{G}^* \right\}  \right)^{-1}  \mb{G} \Fbb \right| \\
	& =  \log_2 \left| \bI +  \frac{1 - \eta_b}{\eta_b} \mb{A}^* \mr{diag}\left\{ \frac{1}{ ||\mb{a}_i|| }\right\} \mr{diag}\left\{ \frac{1}{ ||\mb{a}_i|| }\right\} \mb{A} \right| \\
	& =  \log_2 \left| \bI +  \frac{1 - \eta_b}{\eta_b} \widetilde{\mb{A}}^* \widetilde{\mb{A}} \right| \\
	& =  \sum_{i=1}^{\Ns} \log_2 \left(1 + \frac{1-\eta_b}{\eta_b} \lambda_i \left( \widetilde{\mb{A}}^* \widetilde{\mb{A}} \right)\right)
	\end{align}
	where $\bA \triangleq \mb{G} \Fbb$ is a Hermitian matrix, $\ba_i$ is the $i$-th column of $\bA$, and $\widetilde{\bA}$ is obtained by normalizing each each row of $\bA$.
	Since each row of $\widetilde{\bA}$ has unit norm, then
	\begin{align}
	\sum_{i=1}^{\Ns} \lambda_i \left( \widetilde{\bA}^* \widetilde{\bA} \right) = \mr{tr} \left( \widetilde{\bA}^* \widetilde{\bA} \right) = \Ns.
	\end{align}
	Therefore, we have
	\begin{align}
	R_{\mr{AQNM}} & \stackrel{(a)}{\leq}  \Ns \log_2 \left( 1 + \frac{1-\eta_b}{\eta_b} \right) \\
	&= \Ns \log_2 \frac{1}{\eta_b},
	\end{align}
	where $(a)$ follows from Jensen's inequality and that $\log_2 (1+x) $ is concave in $x$. When the ADC resolution $b$ is large ($b \geq 3$), the distortion factor $\eta_b$ can be approximated as \cite{Gersho_Book12}
	\begin{align}
		\eta_b \approx \frac{\pi \sqrt{3}}{2} 2^{-2b}.
	\end{align}
    As a result, the rate obtained by AQNM model is upper bounded by
    \begin{align}
      R_{\mr{AQNM}} &\leq 2 \Ns b - \Ns \log_2 \frac{\pi \sqrt{3}} {2} \\
      &\approx  2 \Ns b - 1.44\Ns .
    \end{align}
    However, we know the achievable rate of quantized MIMO channel is upper bounded by $2 \Ns b$ bps/Hz  and the channel inversion method can achieve the bound at high enough SNR as shown in \eqref{eq:R_ci_lb_bbit}.
    Therefore, the AQNM is not an accurate model at high SNR.
    The reason is threefold.
    First, the input signal $\bs$ is assumed to be follow suboptimal continuous Gaussian. Second, the quantization noise is assumed to be the worst-case Gaussian noise. Third, the Max-Loyd quantizer minimizing the MSE is not necessarily optimum for maximizing the channel capacity.

	\section{Simulation Results}
	
	We evaluate the performance of proposed methods in a mmWave MIMO channel with large antenna arrays and limited number of transmit and receive RF chains. According to measurement results \cite{Akdeniz_JSAC14, Rappaport_TCOM15}, the number of clusters tends to be lower in the mmWave band compared with lower frequencies. The mmWave channel will mostly consist of the line-of-sight (LOS) and a few NLOS clusters. In the simulations, the wireless channel is assumed to have 4 clusters, each of which consists of 5 rays. The angle spread is $7.5$ degrees. These numbers are chosen according to the urban macro (UMa) NLOS channel measurement results at 28 GHz given in the white paper \cite{WP_5G_Channel_Model}.
	The results are obtained by averaging over 100 channel realizations.
	
	\subsection{Achievable Rates}
\begin{figure}[]
    \centering
    \begin{subfigure}[]{0.50\columnwidth}
        \centering
        \includegraphics[trim= 0 0 20 10,clip,width=1\columnwidth]{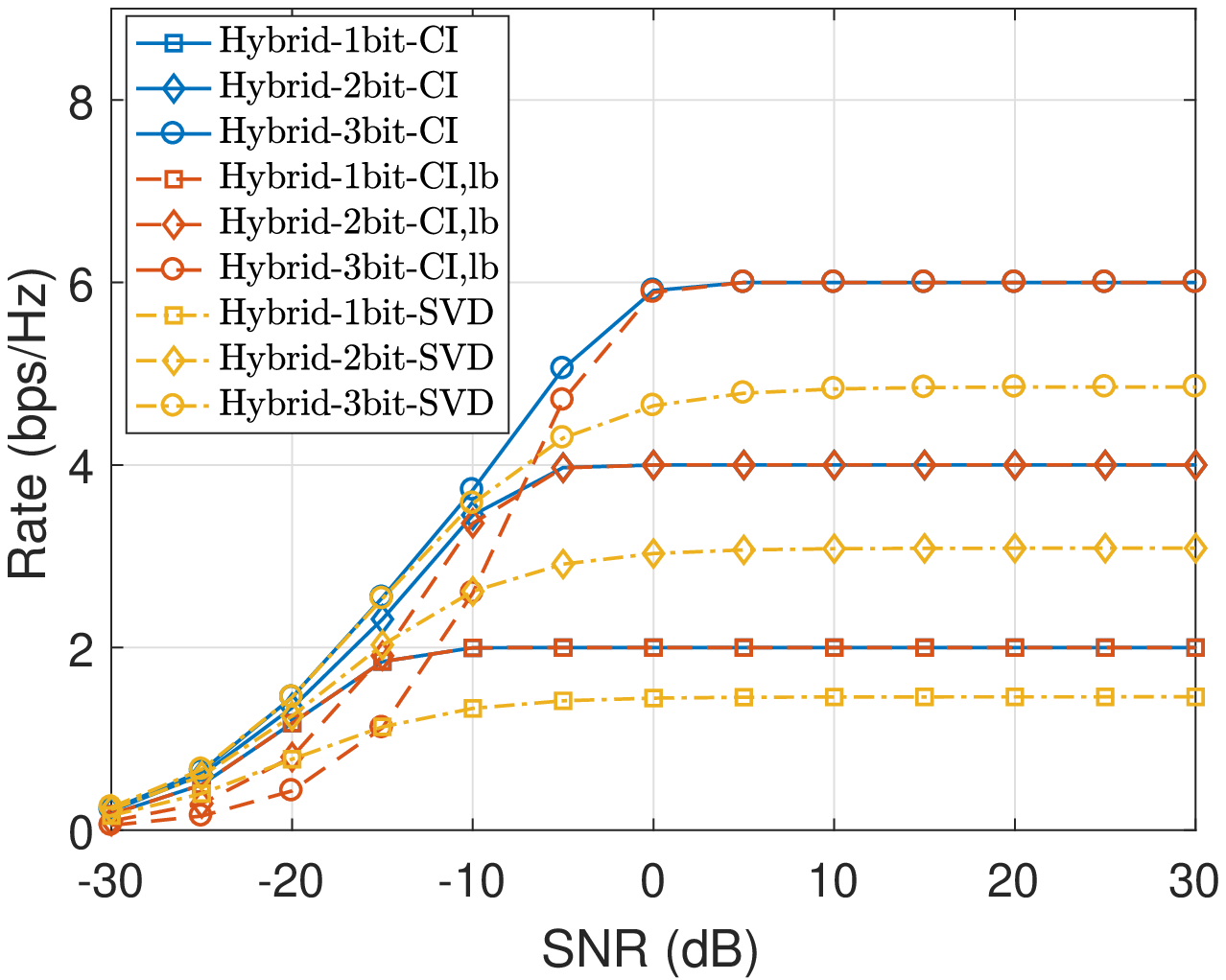}
        \vspace{-1cm}
        \caption{$\Ncr=1$}
        \label{fig:Rate_vs_SNR_Nt_64_Nr_8_Nct_8_Ncr_1_Clustered_CI_SVD}
    \end{subfigure}%
    \begin{subfigure}[]{0.50\columnwidth}
        \centering
        \includegraphics[trim= 0 0 20 10,clip,width=1\columnwidth]{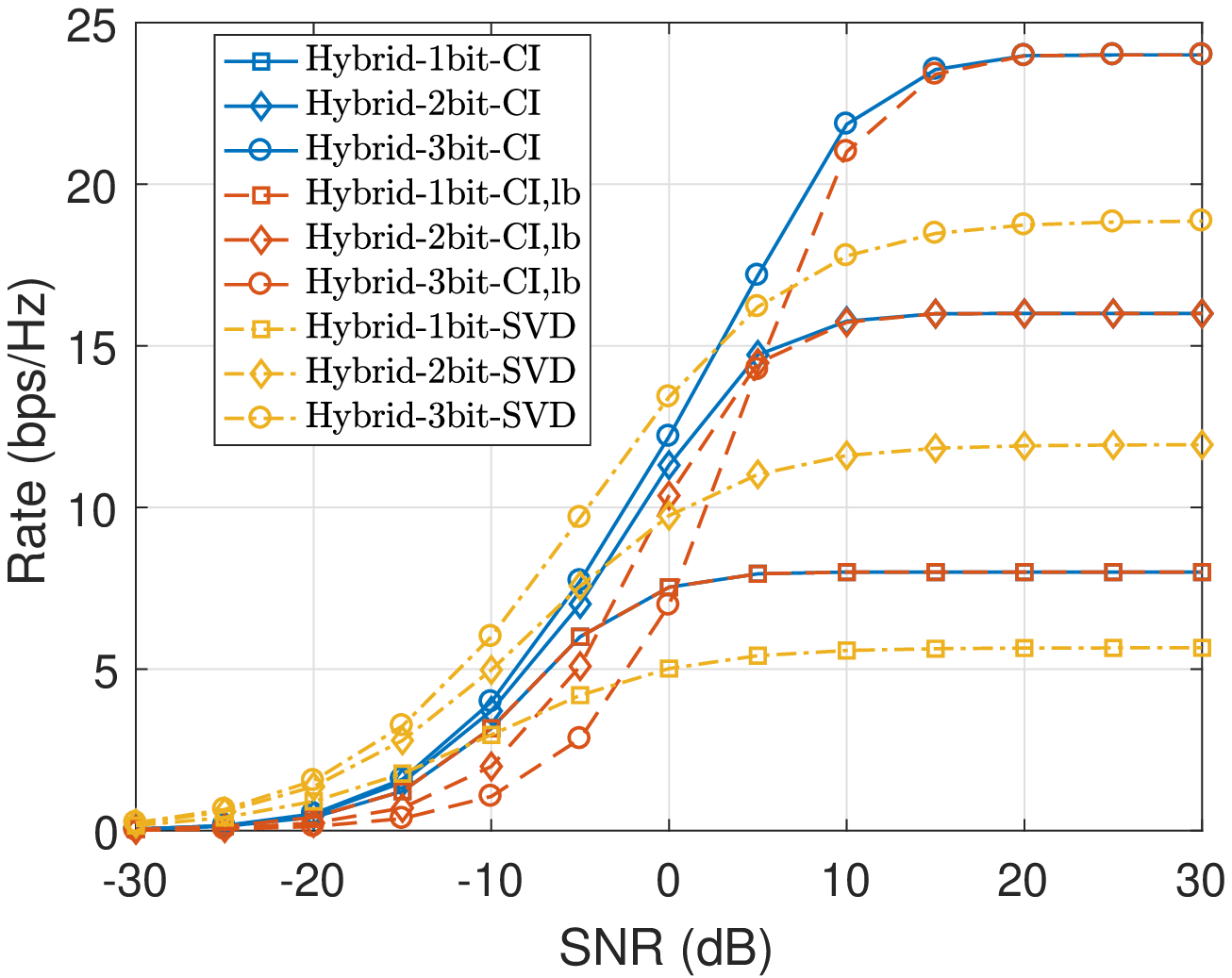}
        \vspace{-1cm}
        \caption{$\Ncr=4$}
        \label{fig:Rate_vs_SNR_Nt_64_Nr_8_Nct_8_Ncr_4_Clustered_CI_SVD}
    \end{subfigure}
    \caption{This figure shows rates versus SNR of different transmission methods. The transmitter is assumed to have  $64$ antennas and $8$ RF chains. The receiver employs $8$ antennas. In Fig. (a), the receiver has $1$ RF chains while in Fig. (b), the receiver has $4$ RF chains.}
    \label{fig:Rate_vs_SNR_Nt_64_Nr_8_Nct_8_Clustered_CI_SVD}
    \vspace{-0.3cm}
\end{figure}

	In this subsection, we evaluate the achievable rates of the proposed architecture by numerical simulations. The transmitter is assumed to have  $\Nt=64$ antennas and $\Nct = 8$ RF chains, while the receiver employs $\Nr=8$ antennas. The number of data steams is assumed to same as the number of receive RF chains, i.e., $\Ns = \Ncr$.
	
	\figref{fig:Rate_vs_SNR_Nt_64_Nr_8_Nct_8_Clustered_CI_SVD} shows the achievable rates when the receiver has 1 and 4 RF chains, respectively. In Fig. \ref{fig:Rate_vs_SNR_Nt_64_Nr_8_Nct_8_Ncr_1_Clustered_CI_SVD}, it is seen that when there is only one RF chains at the transmitter, the channel inversion and SVD method has close performance at low SNR.
	At high SNR, however, the channel inversion method achieves the rate $2b$ bps/Hz while the rate of SVD method saturates to $\log_2 \frac{1}{\eta_b}$ which is about $1.46$, $3.09$, $4.86$ bps/Hz for $b=1$, $2$, $3$, respectively. In \figref{fig:Rate_vs_SNR_Nt_64_Nr_8_Nct_8_Ncr_4_Clustered_CI_SVD}, another case when $\Ncr=4$ is shown. It is found that although at high SNR the channel inversion method achieves larger rate than SVD method, its performance at low SNR is much worse. The reason is that since the channel has only 4 clusters, the fourth largest singular value $\nu_4$ is small and the power loss $\log_2 \frac{\nu_1^2}{\nu_4^2}$ is large. Last, the lower bound provided in \eqref{eq:R_ci_lb_bbit} is also plotted. It is seen that the lower bound is tight for 1-bit ADC case. For other cases, the lower bound is tight at high SNR.
	
\begin{figure}[]
	\centering
	\begin{subfigure}[]{0.5\columnwidth}
		\centering
		\includegraphics[trim= 0 0 20 10,clip,width=1\columnwidth]{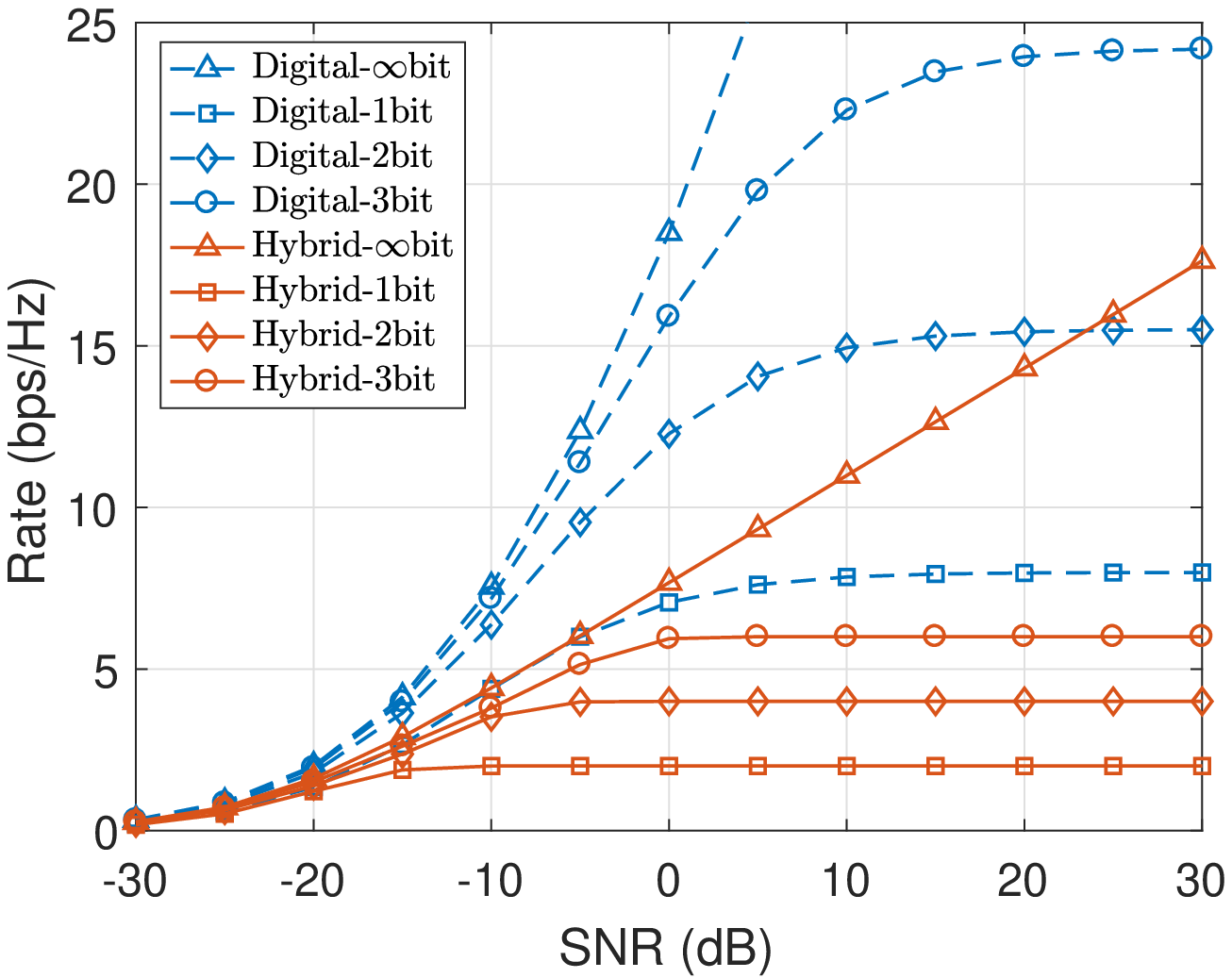}
		\vspace{-1cm}
		\caption{$\Ncr=1$}
		\label{fig:Rate_vs_SNR_Nt_64_Nr_8_Nct_8_Ncr_1_Clustered_Digital_Hybrid}
	\end{subfigure}%
	\begin{subfigure}[]{0.5\columnwidth}
		\centering
		\includegraphics[trim= 0 0 20 10,clip,width=1\columnwidth]{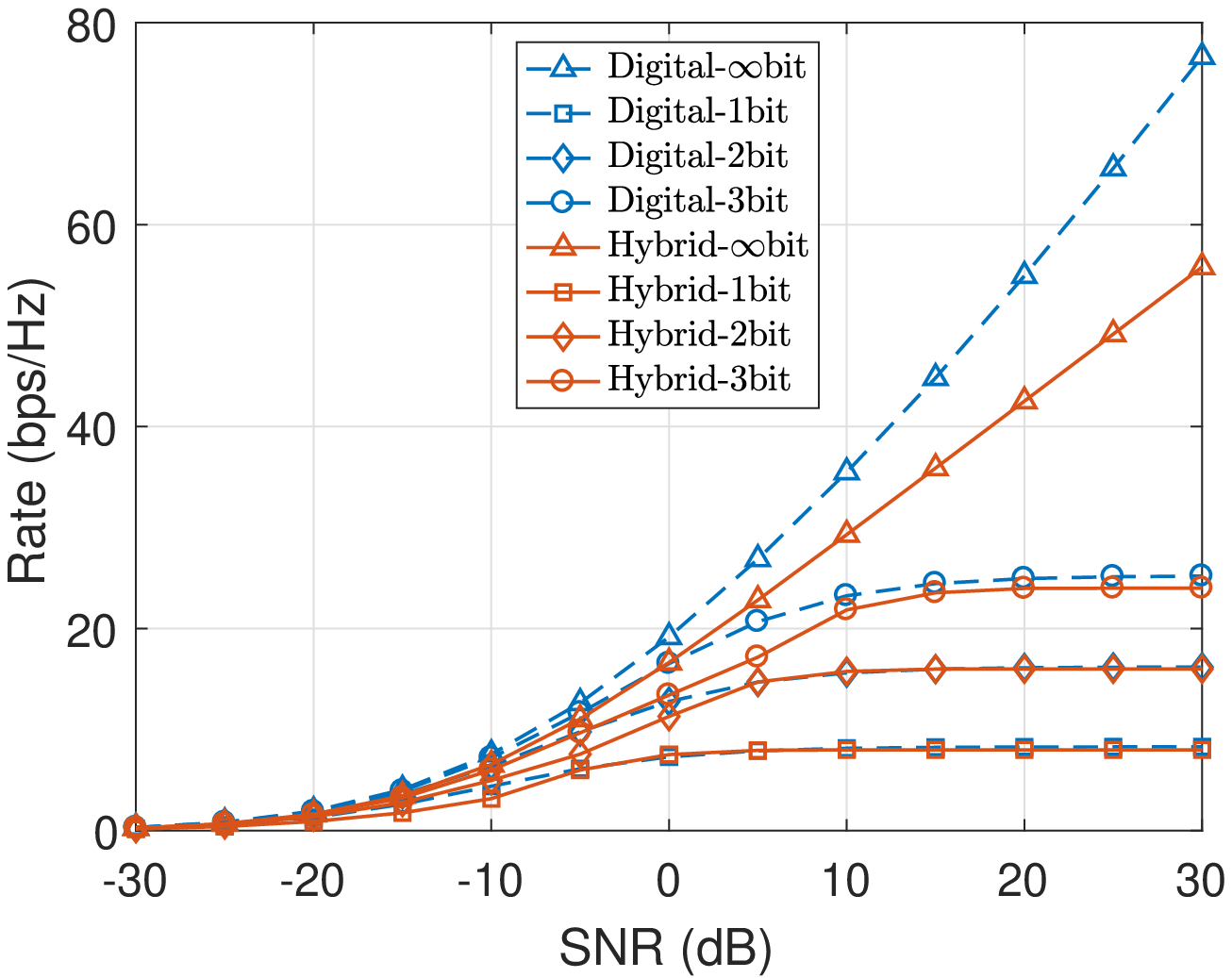}
		\vspace{-1cm}
		\caption{$\Ncr=4$}
		\label{fig:Rate_vs_SNR_Nt_64_Nr_8_Nct_8_Ncr_4_Clustered_Digital_Hybrid}
	\end{subfigure}
	\caption{This figure shows rates versus SNR of different transmission methods.
The transmitter is assumed to have  $64$ antennas and $8$ RF chains.
The receiver employs $8$ antennas. In Fig. (a), the receiver has $1$ RF chains while in Fig. (b), the receiver has $4$ RF chains.}
	\label{fig:Rate_vs_SNR_Nt_64_Nr_8_Nct_8_Clustered_Digital_Hybrid}
	\vspace{-0.3cm}
\end{figure}

\figref{fig:Rate_vs_SNR_Nt_64_Nr_8_Nct_8_Clustered_Digital_Hybrid} compares the achievable rate of fully-digital and hybrid architecture. The rate of hybrid architecture is the maximum of the CI and SVD method. The gap between the curves of ``Digital-$b$ bit'' and ``Hybrid-$b$ bit'' represents the loss due to limited number of RF chains while the gap between the curves of ``$\infty$bit'' and ``$b$bit'' is the loss due to low resolution ADCs.
It is seen that the digital architecture is much better than the hybrid architecture with only one RF chain because the hybrid architecture can only support single stream transmission while the digital architecture can support at most 8 streams simultaneously. However, when there are 4 receive RF chains in the hybrid architecture, the gap between these two architecture is small since there are the channel has 4 clusters. Last, we can see the loss due to low resolution ADCs is small at low and medium SNRs. For example, the gap between the curve ``Hybrid-$\infty$bit'' and ``Hybrid-3bit'' is less than 3 dB when the SNR is less than 10 dB in \figref{fig:Rate_vs_SNR_Nt_64_Nr_8_Nct_8_Ncr_4_Clustered_Digital_Hybrid}.



\begin{figure}[]
	\centering
	\begin{subfigure}[]{0.5\columnwidth}
		\centering
		\includegraphics[trim= 0 0 20 10,clip,width=1\columnwidth]{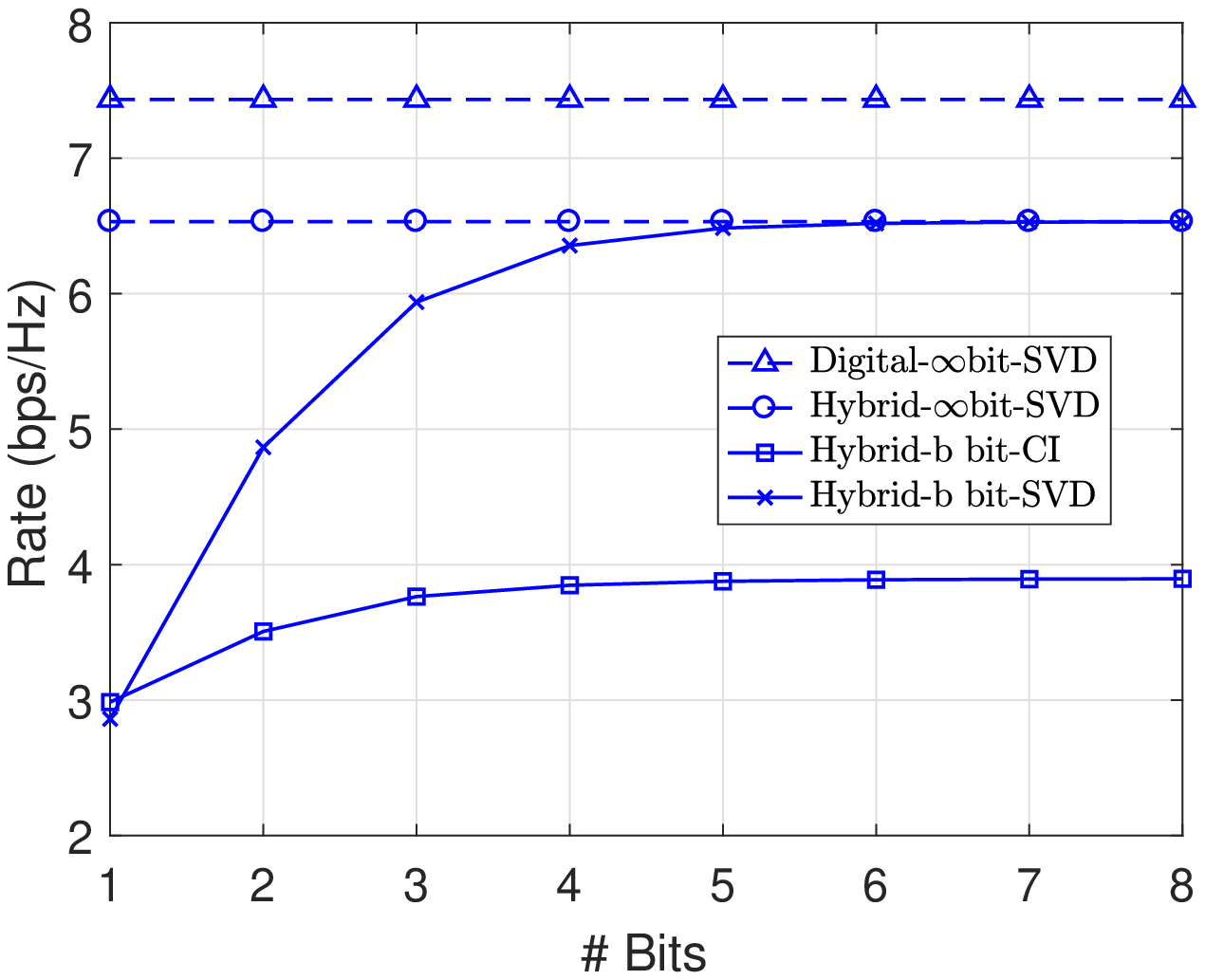}
		\vspace{-1cm}
		\caption{SNR$=-10$ dB}
		\label{fig:Rate_vs_Bits_Nt_64_Nr_8_Nct_8_Ncr_4_SNR_-10dB_Clustered}
	\end{subfigure}%
	\begin{subfigure}[]{0.5\columnwidth}
		\centering
		\includegraphics[trim= 0 0 20 10,clip,width=1\columnwidth]{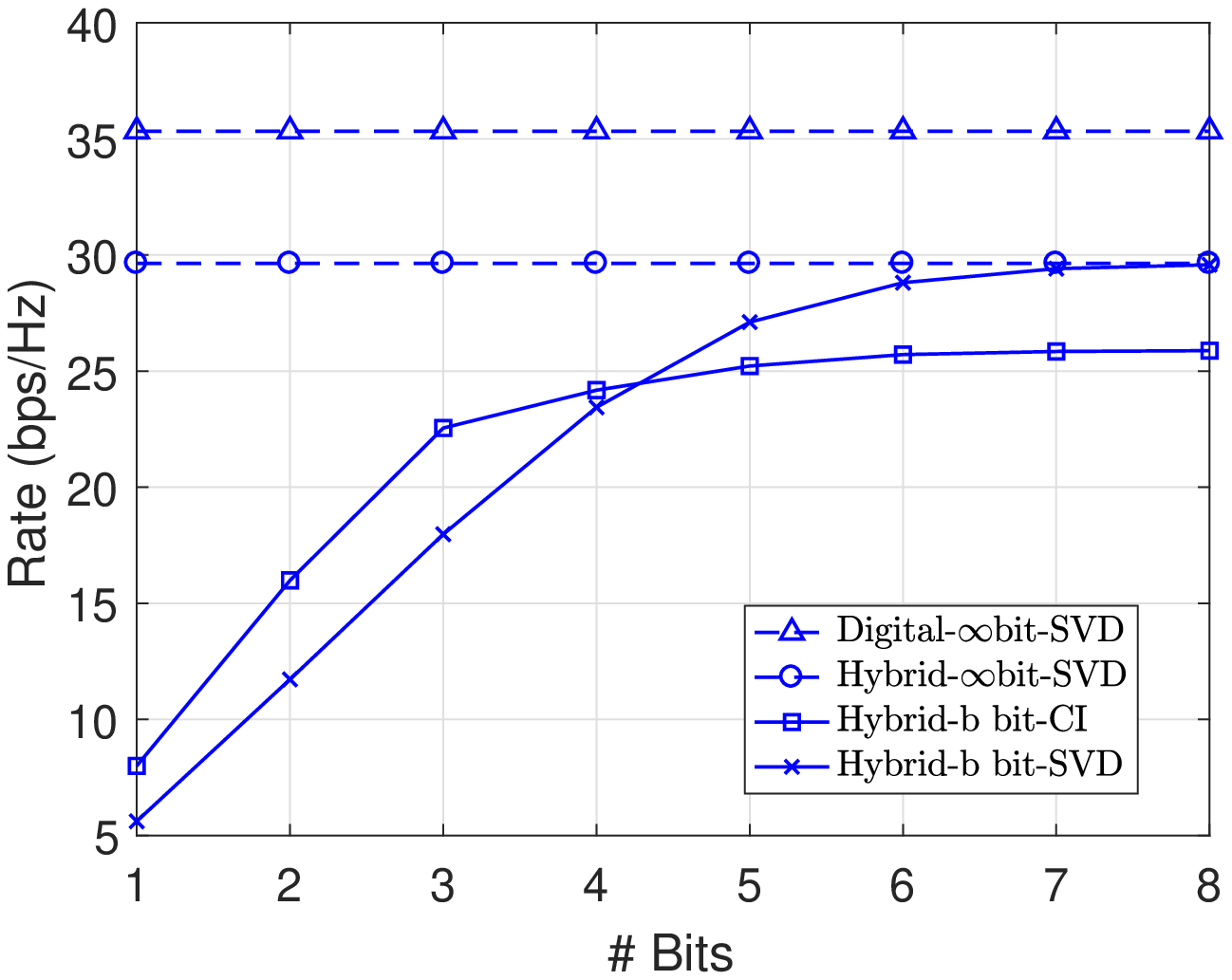}
		\vspace{-1cm}
		\caption{SNR$=10$ dB}
		\label{fig:Rate_vs_Bits_Nt_64_Nr_8_Nct_8_Ncr_4_SNR_10dB_Clustered}
	\end{subfigure}
	\caption{This figure shows rates versus ADC resolution for different transmission methods. The transmitter is assumed to have  $64$ antennas and $8$ RF chains. The receiver employs $8$ antennas and $4$ RF chains. In Fig. (a), the SNR is $-10$ dB while in Fig. (b), the SNR is $10$ dB.}
	\label{fig:Rate_vs_Bits_Nt_64_Nr_8_Nct_8_Ncr_4_Clustered}
	\vspace{-0.3cm}
\end{figure}
	
	Fig. \ref{fig:Rate_vs_Bits_Nt_64_Nr_8_Nct_8_Ncr_4_Clustered} shows the achievable rate with respect to the ADC resolution. First, as expected, the rates of the finite-bit ADC receiver increase with resolution. Second, with multi-bit ADCs (5-bit when SNR = $-10$ dB and 7-bit when SNR = $10$ dB), the SVD method achieves the performance similar to that of hybrid architecture with $\infty$-bit ADCs. This implies that high resolution ADC do not provide much gain compared to the few-bit ADC when the SNR is low.
	Third, when the ADC resolution is low, channel inversion method is better than the SVD method while with high resolution quantization, the SVD method is better. This is reasonable since with high resolution ADC, the channel is close to the one without quantization and in a unquantized channel, SVD method is optimum.


	 \figref{fig:Rate_vs_Ncr_Nt_64_Nr_8_Nct_8_Clustered} presents the achievable rates versus the number of RF chains at the receiver. First, we find that the rate of SVD method always increases with $\Ncr$. Second, at low SNR ($-10$ dB), the channel inversion method achieve the largest rate when $\Nct=2$. This means at low SNR, it is better to turn off some RF chains and transmit fewer number of steams. Note that the power consumption also decreases by turning off some RF chains.
	Third, we find that compared to fully-digital architecture where $\Nct=\Nt=64$ and $\Ncr=\Nr=8$, the hybrid architecture with limited number of RF chains ($\Nct=8$ and $\Ncr=2$) and low resolution ADCs (4-bit) incurs about 20\%-30\% spectral efficiency loss. As shown in the next subsection, however, the energy efficiency of the proposed receiver is much higher than the fully-digital architecture.
	
	\begin{figure}[]
		\centering
		\begin{subfigure}[]{0.5\columnwidth}
			\centering
			\includegraphics[trim= 0 0 20 5,clip,width=1\columnwidth]{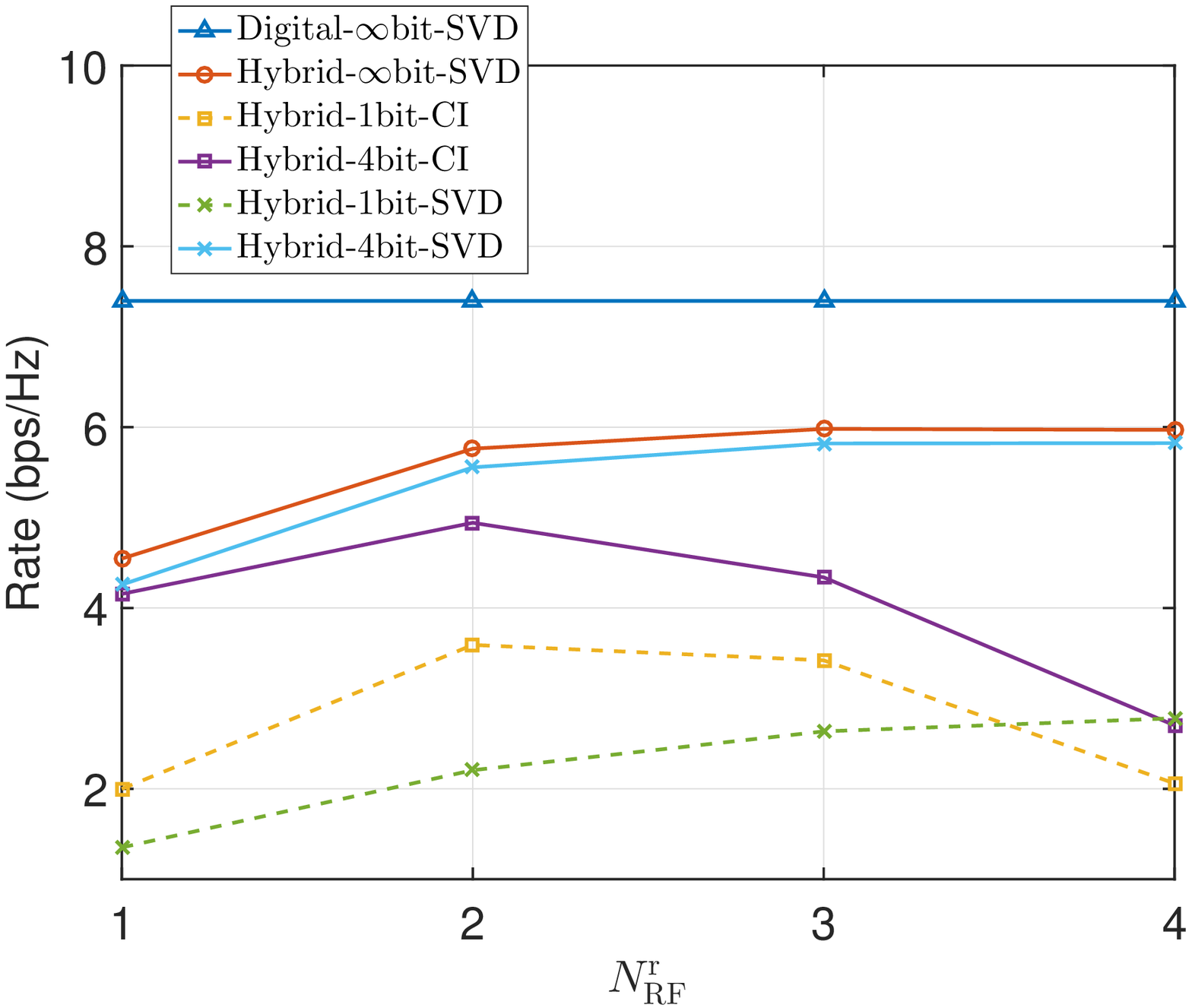}
			\vspace{-1cm}
			\caption{SNR$=-10$ dB}
			\label{fig:Rate_vs_Ncr_Nt_64_Nr_8_Nct_8_SNR_-10dB_Clustered}
		\end{subfigure}%
		\begin{subfigure}[]{0.5\columnwidth}
			\centering
			\includegraphics[trim= 0 0 20 5,clip,width=1\columnwidth]{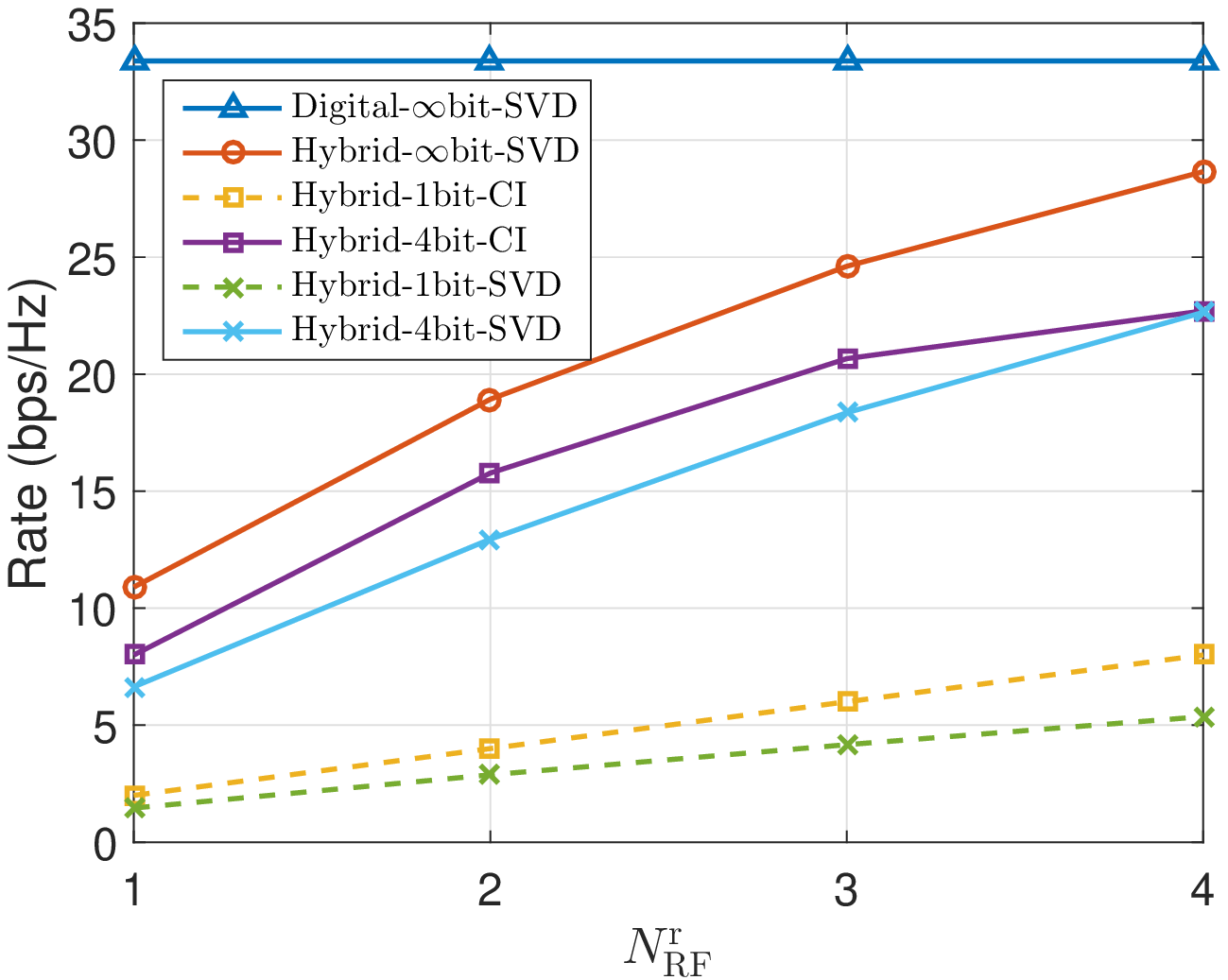}
			\vspace{-1cm}
			\caption{SNR$=10$ dB}
			\label{fig:Rate_vs_Ncr_Nt_64_Nr_8_Nct_8_SNR_10dB_Clustered}
		\end{subfigure}
		\caption{This figure shows rates versus the number of receive RF chains for different transceiver architecture and transmission methods. The transmitter is assumed to have  $64$ antennas and $8$ RF chains, while the receiver employs $8$ antennas. In Fig. (a), the SNR is $-10$ dB while in Fig. (b), the SNR is $10$ dB.}
		\label{fig:Rate_vs_Ncr_Nt_64_Nr_8_Nct_8_Clustered}
		\vspace{-0.3cm}
	\end{figure}

	\subsection{Energy efficiency}
	
	In this subsection, we evaluate the performance of the different receiver architectures by investigating the trade-off between their achievable rates and power consumption. First, we formulate a generic power consumption model for the hybrid architecture with low-resolution ADCs. Then, we use this model in the performance evaluation. Consider the system model in \figref{fig:System_model} with a $b$-bit ADC receiver having $\Nr$ antennas and $\Ncr$ RF chains. Let $P_\mathrm{LNA}$, $P_\mathrm{PS}$, $P_\mathrm{RFchain}$, $P_\mathrm{ADC}$, $P_\mathrm{BB}$ denote the power consumption in the LNA, phase shifter, RF chains, ADC, and baseband processor, respectively. Then, the consumed power by the hybrid combining receiver in \figref{fig:System_model} can be approximated as \cite{Mendez-Rial_Access16}
	\begin{equation}\label{eq:PowerModel1}
	P_\mathrm{tot}= \Nr P_\mathrm{LNA}+\Ncr\left(\Nr P_\mathrm{PS}+P_\mathrm{RF chain}+2 P_\mathrm{ADC} \right)+P_\mathrm{BB},
	\end{equation}
	where the power consumed by ADCs can be further expressed in terms of the number of bits as
	%
	\begin{equation}\label{eq:PowerModel2}
	P_\mathrm{ADC}= FOM_{W} \cdot f_s \cdot 2^{b},
	\end{equation}
	where $f_s$ is the Nyquist sampling rate, $b$ is the number of bits, and $FOM_{W}$ is Walden's figure-of-merit for evaluating ADC's power efficiency with resolution and speed \cite{Walden_JSAC99,Murmann_16},

	\begin{figure}[t]
	\begin{centering}
		\includegraphics[width=.7\columnwidth]{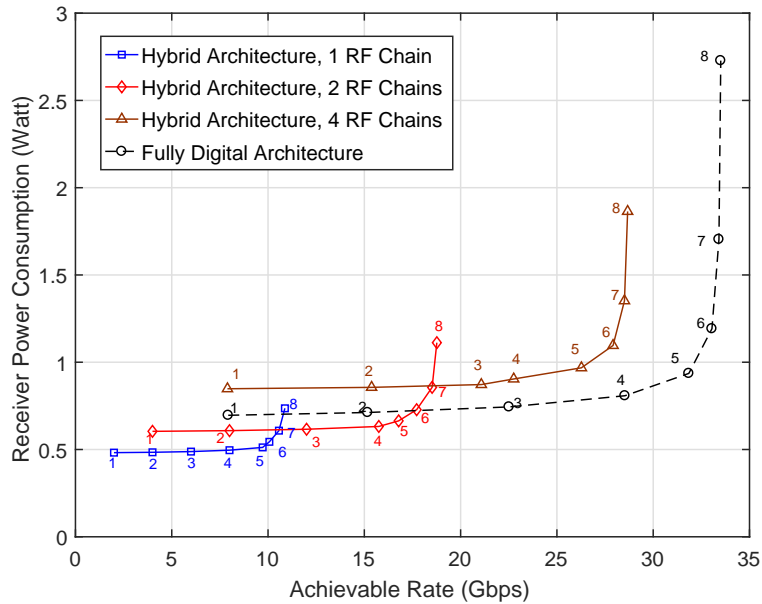}
		\vspace{-0.1cm}
		\centering
		\caption{Trade-off between the achievable rate and power consumption for hybrid combiners with different number of RF chains and ADC bits. The results also include the fully-digital architecture with finite-bit ADCs. The transmitter is assumed to have  $64$ antennas and $8$ RF chains, while the receiver employs $8$ antennas and  $1, 2 \ \text{or} \ 4$ RF chains. The SNR is 10 dB and the bandwidth is 1 GHz.}
		\label{fig:Power_Rate}
	\end{centering}
	\vspace{-0.3cm}
	\end{figure}
	
	Next, we present some simulation results in Figs. \ref{fig:Power_Rate}-\ref{fig:EE_vs_Bits} that illustrate the power consumption-rate trade-off of the hybrid combining receiver for different number of RF chains and ADC quantization bits at the receiver. In these simulations, we adopt the same clustered channel model, and consider the system model in \sref{sec:System_Model} with a transmitter having $64$ antennas and $8$ RF chains, and a receiver deploying $8$ antennas.
	
	First, we plot in \figref{fig:Power_Rate} the power consumption versus the achievable rate for the fully-digital and proposed hybrid architecture with few-bit ADC receivers. The power consumption in the different receivers are calculated based on the power model in \eqref{eq:PowerModel1}-\eqref{eq:PowerModel2}. The consumed power in the RF components are assumed to be $ P_\mathrm{LNA}=20$ mW, $P_\mathrm{PS}=10$ mW, $P_\mathrm{RF chain}=40$ mW, $P_\mathrm{BB}=200$ mW \cite{Mendez-Rial_Access16, Lin_Yu-Hsuan_IMS16}. 
	For the ADCs, we adopt the power consumption model in \eqref{eq:PowerModel2} with $FOM_{W}=500 $ fJ/conversion-step, which is a typical achievable value at 1 GHz \cite{Murmann_16,Ali_JSSC14} \footnote{The minimum achievable value of $FOM_{W}$ can be as low as $5$ fJ/conversion-step
		 as shown in the Figure ``Walden FOM vs. Speed'' in \cite{Murmann_16}. }.

	The achievable rate for the hybrid architecture with few-bit ADCs in \figref{fig:Power_Rate} is calculated as the maximum of the two achievable rate expressions in \eqref{eq:R_ci_bbit} and \eqref{eq:R_AQNM_lb} for CI and SVD based precoding. Several important insights can be obtained from \figref{fig:Power_Rate}.
	First, for all the four cases in the figure, the achievable rate significantly increases when the number of bits increases from $1$ to $4-5$ bits, with a negligible increase in the receiver power consumption. This implies that with these system configurations, having $4-5$ bits ADC's can lead to a better trade-off between the rate and power consumption compared with having 1-bit ADC receivers. Next, going from $4-5$ bits to $7-8$ bits for the ADCs in the hybrid architectures slightly improves the achievable spectral efficiency, but with significantly more power consumption. This means that $4-5$ ADC bits can be better than $7-8$ ADC bits when investigating the rate-power consumption trade-off.
	Second, it is found that among the four cases considered, the hybrid one with 1 RF chain gives the best power-rate trade-off when the rate is less than 11 Gbps, the hybrid one with 2 RF chains gives the best trade-off when the rate is between 11 Gbps and 18 Gbps, and the fully-digital architecture is best when the rate is larger than 18 Gbps.
	%
	%
	Third, when there are 4 RF chains, the performance of hybrid architecture is dominated by that of the fully-digital architecture. To achieve the same rate, the power consumption of hybrid receiver with 4 RF chains is always larger than of fully-digital one because of the high power consumption of the large number of analog phase shifters.

	\begin{figure}[]
		\centering
		\begin{subfigure}[]{0.5\columnwidth}
			\centering
			\includegraphics[trim= 0 0 20 10,clip, width=1\columnwidth]{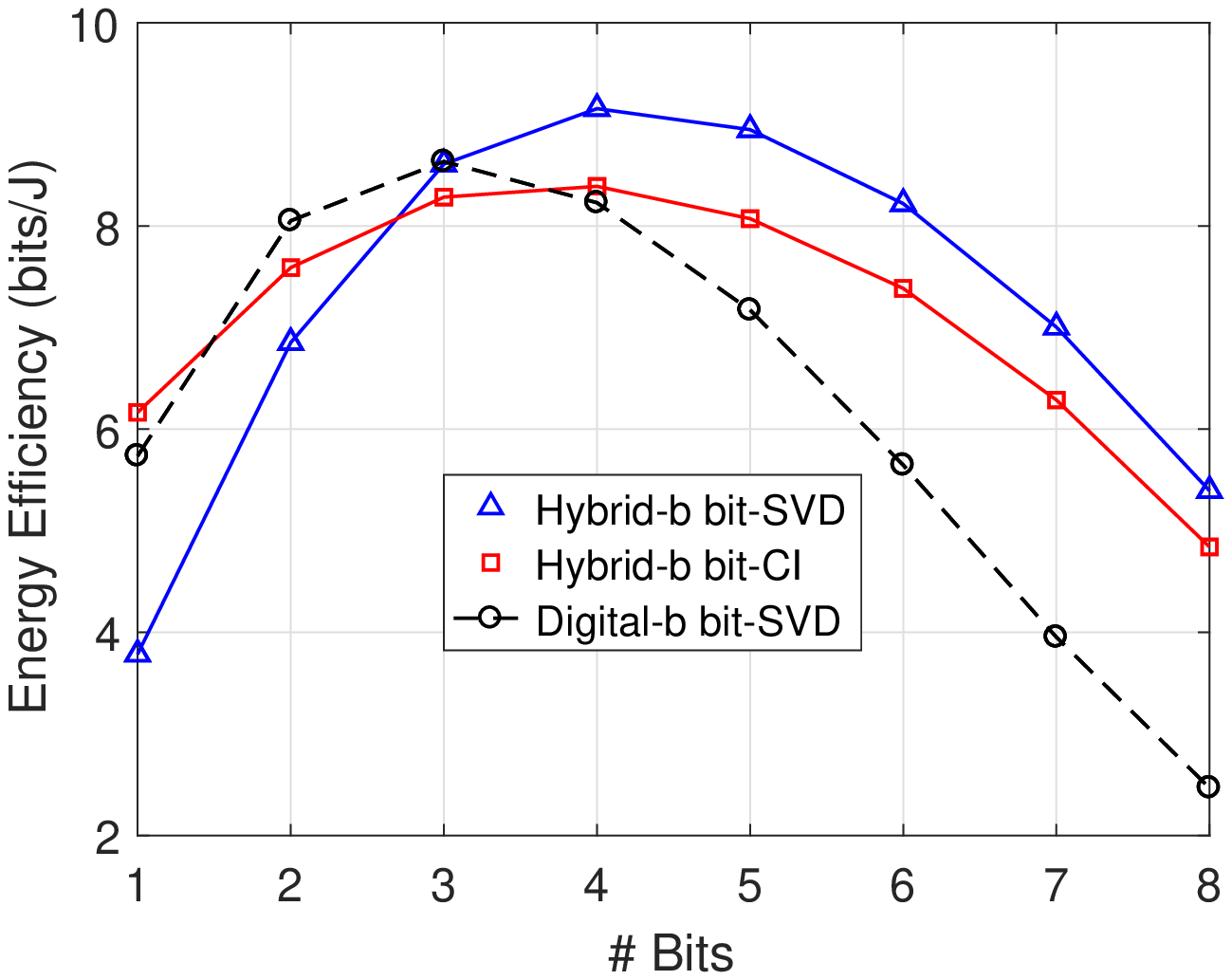}
			\vspace{-1cm}
			\caption{SNR$=-10$ dB}
			\label{fig:EE_-10dB}
		\end{subfigure}%
		\begin{subfigure}[]{0.5\columnwidth}
			\centering
			\includegraphics[trim= 0 0 20 10,clip,width=1\columnwidth]{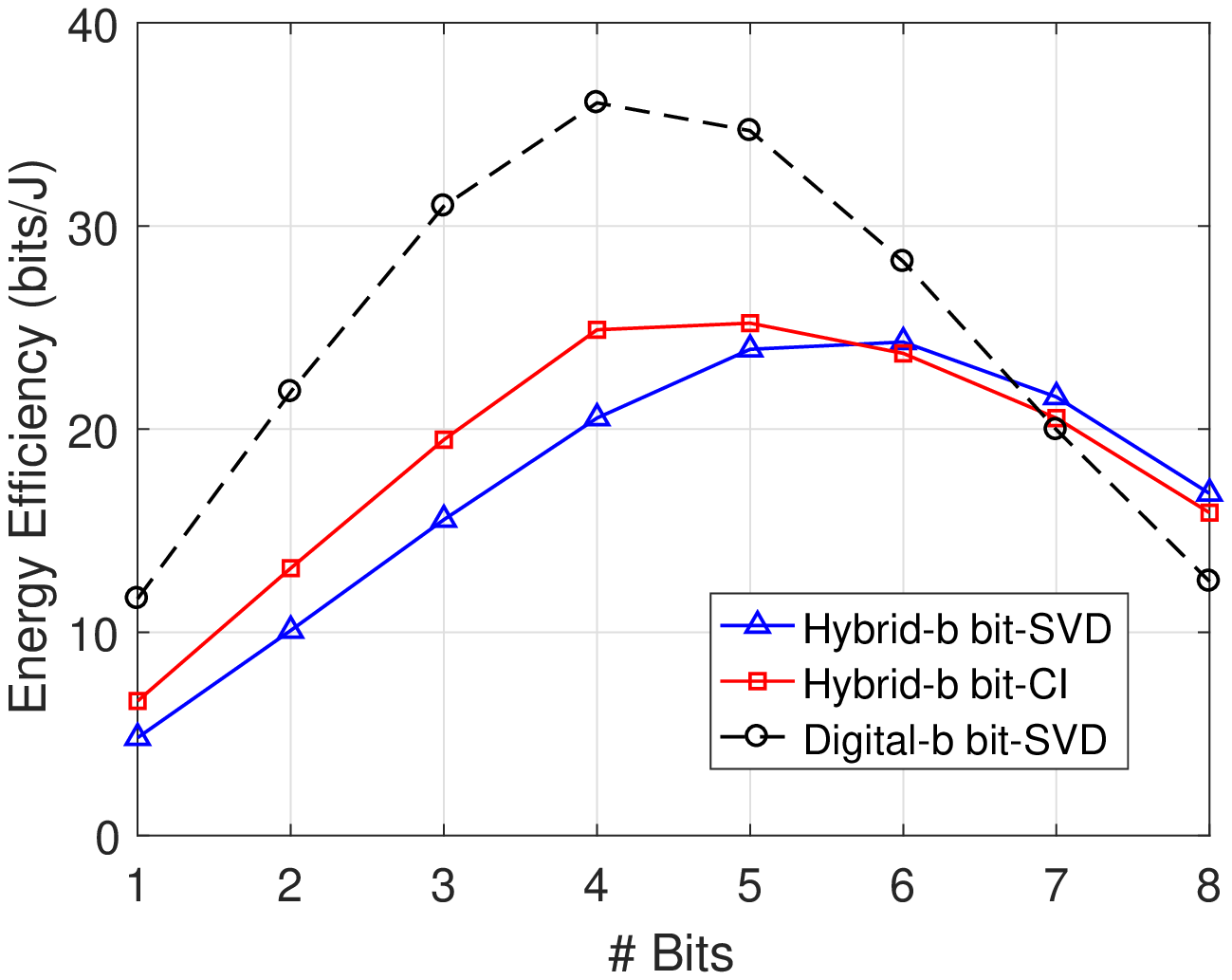}
			\vspace{-1cm}
			\caption{SNR$=10$ dB}
			\label{fig:EE_10dB}
		\end{subfigure}
		\caption{Energy efficiency of the digital and hybrid architecture with few-bit ADC receivers for different numbers of bits. The transmitter is assumed to have  $64$ antennas and $8$ RF chains, while the receiver employs $8$ antennas and  $2$ RF chains. In Fig. (a), the SNR is $-10$ dB while in Fig. (b), and the SNR is $10$ dB.}
		\label{fig:EE_vs_Bits}
		\vspace{-0.3cm}
	\end{figure}
	
	Based on the observations made from \figref{fig:Power_Rate}, we investigate the energy efficiency of the different receiver architectures in \figref{fig:EE_vs_Bits}. The energy spectral efficiency $\eta_\mathrm{EE}$ is defined as
	\begin{equation}
	\eta_\mathrm{EE}=\frac{R \ W}{P_\mathrm{tot}} \mathrm{bits}/ \mathrm{Joule},
	\end{equation}
	where $R$ is the achievable spectral efficiency and $W$ is the transmission bandwidth. The channel and system models adopted in \figref{fig:EE_-10dB}-\figref{fig:EE_10dB} are similar to that in \figref{fig:Power_Rate} with $2$ RF chains used at the receiver and with $W=1$ GHz bandwidth. In \figref{fig:EE_-10dB}, the energy efficiency of the few-bit ADC hybrid architecture is plotted for two different precoding techniques versus the number of ADC bits at SNR$=-10$ dB. The result is compared with the energy efficiencies of the fully-digital transceiver. \figref{fig:EE_-10dB} shows that the number of bits maximizing the energy efficiency is finite and equals to $4$ bits. This maximum energy efficiency is achieved with the SVD precoding in \eqref{eq:R_AQNM_lb}. It is worth noting from \figref{fig:EE_-10dB} that the choice of the precoding technique depends on the number of ADC quantization bits. While the channel inversion achieves a better energy efficiency with $1-2$ ADC bits, SVD precoding can lead to a better energy efficiency for higher numbers of ADC quantization bits. \figref{fig:EE_10dB} has the same comparison in \figref{fig:EE_-10dB} with the only difference of operating at SNR=$10$ dB. The energy efficiency trends in \figref{fig:EE_10dB} are similar to that in \figref{fig:EE_-10dB} with the main difference that the maximum energy efficiency is achieved with the channel inversion precoding in  \eqref{eq:R_ci_bbit}. This implies that determining the precoding techniques should take the SNR into consideration. We note though the $4$ bits for the ADC quantization still achieves the optimal energy efficiency. Further, \figref{fig:EE_-10dB}-\figref{fig:EE_10dB} also show that energy efficiency of the proposed few-bit ADC hybrid architecture can be much better that that of the fully-digital architectures that employ high resolution ADCs. The reason is that the ADC will dominate the power consumption when the resolution is high and hybrid architecture saves power by reducing the number of high-resolution ADCs.
	
	\section{Conclusion}
	
	In this paper, we derived the achievable spectral efficiency and energy-rate trade-off of a generalized hybrid architecture with few-bit ADC receivers. First, we considered channel inversion and SVD precoding based transmission methods and derived their achievable rates. The transmission methods include three elements: the design of analog and digital precoding, the choice of the transmit signal distribution and the setup of quantizer.
	Simulation results showed that at the low and medium SNRs, the proposed architecture and precoding can achieve a comparable rate to the fully-digital solution. Second, we explored the trade-off between the achievable rate and the power consumption, which is particularly important in massive MIMO and mmWave systems. Adopting energy efficiency as a performance metric, the proposed hybrid architecture with few-bit ADC receiver is compared with the fully-digital transceiver and the conventional hybrid architecture with full-resolution ADCs. Numerical results showed that the generalized hybrid architecture with coarse ADC quantization achieves a potential energy efficiency gain compared to the other architectures. The results illustrated that coarse ADC quantization with $4-5$ bits normally achieves the maximum energy efficiency in various system configurations.

	
	
	\appendices
	\section{Proof of Proposition 1}\label{app_upperbound}
	
	The proof is similar to the proof of \cite[Proposition 4]{Mo_Jianhua_TSP15}.
	The channel without any constraint on digital precoding and combining is,
	\begin{align}
	\mb{r} &= \mathrm{sgn} \left( \Wrfx \mb{H} \Frf  \mb{x} + \Wrfx \mb{n} \right).
	\end{align}
	We want to find a upper bound of the mutual information $I(\br; \bx)$ subject to the power constraint $\mathbb{E}[\|\Frf \bx\|^2] \leq \Pt$.
	
	First, we have,
	\begin{align}
	I(\mathbf{r}; \mathbf{x}) &= \cH(\mathbf{r}) - \cH(\mathbf{r}|\mathbf{x}) \\
	& \stackrel{(a)}{\leq} 2 \Ncr - \cH(\mathbf{r}| \mathbf{x}),
	\end{align}
	where inequality $(a)$ follows from that there are at most $2\Ncr$ quantization outputs.
	
	Next, we derive a lower bound of $\cH(\mathbf{r}|\mathbf{x})$.
	For a given transmitted signal $\mathbf{x}=\mathbf{x}'$, denote $\mathbf{z}' = \Wrfx \mb{H} \Frf \mathbf{x}'$ and $\widetilde{\mathbf{z}}' = [\Re(\mathbf{z}')^T, \Im(\mathbf{z}')^T]^T$. The conditional entropy of $\mathbf{r}$ given $\mathbf{x}=\mathbf{x}'$ is,
	\begin{align}
	\cH( \mathbf{r}| \mathbf{x} = \mathbf{x}')
	&\stackrel{(a)}{=} \sum_{j=1}^{\Ncr} \cH( {\mathbf{r}}_j| {\mathbf{x}'})  \\
	&\stackrel{(b)}{=} \sum_{j=1}^{2 \Ncr} \Hb \left( Q \left( \sqrt{\frac{2}{\sigma_{\rmN}^2}} \widetilde{\mathbf{z}}'_{j} \right) \right)  \\
	&\stackrel{(c)}{=} \sum_{j=1}^{2 \Ncr} \Hb \left( Q \left( \sqrt{\frac{2 \left(\widetilde{\mathbf{z}}'_{j} \right)^2}{\sigma_{\rmN}^2}} \right) \right)
	\end{align}
	where $(a)$ follows from that the effective noise $\Wrfx \bn$ has independent entries, $(b)$ follows from that the in-phase and quadrature parts of the effective noise are independent and Gaussian distributed with variance $\frac{\sigma_{\rmN}^2}{2}$, $(c)$ follows from that $\Hb(Q(x))$ is an even function of $x$.
	Next, noting that
	\begin{align}
	\sum_{j=1}^{2 \Ncr} \left(\widetilde{z}'_j \right)^2 = \|\mathbf{z}'\|^2 \leq \|\mathbf{x}'\|^2 \nu_{1}^2,
	\end{align}
	where $\nu_{1}$ is the maximum singular value of the effective channel $\Wrfx \mb{H} \Frf$,
	we have
	\begin{align}
	\cH( \mathbf{r}| \mathbf{x} = \mathbf{x}') \geq  2 \Ncr \Hb \left( Q \left( \sqrt{ \frac{ \|\mathbf{x}'\|^2 \nu_{1}^2 }{ \Ncr \sigma_{\rmN}^2 } } \right) \right)
	\end{align}
	by Jensen's inequality since $\Hb\left( Q\left( \sqrt{x}\right)\right)$ is decreasing and convex in $x$ (see \cite{Dabeer_SPAWC06} for the proof of convexity).
	
	The conditional entropy of $\mathbf{r}$ is,
	\begin{align}
	\cH({\mathbf{r}}|{\mathbf{x}}) = \mathbb{E}_{\mathbf{x}} \left[ 2 \Ncr \Hb \left( Q \left( \sqrt{ \frac{ \|\mathbf{x}\|^2 \nu_{1}^2}{ \Ncr \sigma_{\rmN}^2 } } \right) \right) \right].
	\end{align}
	
	Since $\Hb(Q(\sqrt{z}))$ is convex, then
	\begin{align}
	\cH(\mathbf{r}|\mathbf{x}) \geq  2 \Ncr \Hb \left( Q \left( \sqrt{ \frac{ \mathbb{E}[\|\mathbf{x}\|^2] \nu_{1}^2}{ \Ncr \sigma_{\rmN}^2 } } \right) \right).
    \end{align}

    We want to minimize $\cH(\mathbf{r}|\mathbf{x})$ subject to the power constraint $\mathbb{E} \left[ \| \Frf \bx \|^2 \right] \leq \Pt$,
	which is equivalent to $\mathbb{E}\left[ \|\mathbf{x} \|^2 \right] \leq \Pt$ by the assumption $\Frfx \Frf = \bI$. As $\Hb(Q(\sqrt{z}))$ is decreasing in $z$, we have
    \begin{align}
    \cH(\mathbf{r}|\mathbf{x})
	& \geq  2 \Ncr \Hb \left( Q \left( \sqrt{ \frac{ \rho \nu_{1}^2 }{ \Ncr} } \right) \right),
	\end{align}
    where $\rho \triangleq \frac{\Pt}{\sigma_{\rmN}^2}$.
	Therefore, the upper bound in \eqref{eq:ub_tight} is obtained.
	
	\linespread{1.3}
	\bibliographystyle{IEEEtran}
	\bibliography{IEEEabrv,One_bit_Quantization,AlkhateebRef_Oct}

\begin{thebibliography}{10}
\providecommand{\url}[1]{#1}
\csname url@samestyle\endcsname
\providecommand{\newblock}{\relax}
\providecommand{\bibinfo}[2]{#2}
\providecommand{\BIBentrySTDinterwordspacing}{\spaceskip=0pt\relax}
\providecommand{\BIBentryALTinterwordstretchfactor}{4}
\providecommand{\BIBentryALTinterwordspacing}{\spaceskip=\fontdimen2\font plus
\BIBentryALTinterwordstretchfactor\fontdimen3\font minus
  \fontdimen4\font\relax}
\providecommand{\BIBforeignlanguage}[2]{{%
\expandafter\ifx\csname l@#1\endcsname\relax
\typeout{** WARNING: IEEEtran.bst: No hyphenation pattern has been}%
\typeout{** loaded for the language `#1'. Using the pattern for}%
\typeout{** the default language instead.}%
\else
\language=\csname l@#1\endcsname
\fi
#2}}
\providecommand{\BIBdecl}{\relax}
\BIBdecl

\bibitem{Mo_Jianhua_WSA16}
J.~Mo, A.~Alkhateeb, S.~Abu-Surra, and R.~W. Heath~Jr., ``Achievable rates of
  hybrid architectures with few-bit {ADC} receivers,'' in \emph{Proceeedings of
  20th International ITG Workshop on Smart Antennas}, March 2016, pp. 1--8.

\bibitem{Marzetta2010}
T.~Marzetta, ``Noncooperative cellular wireless with unlimited numbers of base
  station antennas,'' \emph{IEEE Transactions on Wireless Communications},
  vol.~9, no.~11, pp. 3590--3600, November 2010.

\bibitem{Boccardi2014}
F.~Boccardi, R.~Heath, A.~Lozano, T.~Marzetta, and P.~Popovski, ``Five
  disruptive technology directions for {5G},'' \emph{IEEE Communications
  Magazine}, vol.~52, no.~2, pp. 74--80, Feb. 2014.

\bibitem{Larsson2014}
E.~Larsson, O.~Edfors, F.~Tufvesson, and T.~Marzetta, ``Massive {MIMO} for next
  generation wireless systems,'' \emph{IEEE Communications Magazine}, vol.~52,
  no.~2, pp. 186--195, Feb. 2014.

\bibitem{Rappaport2014}
T.~S. Rappaport, R.~W. Heath~Jr, R.~C. Daniels, and J.~N. Murdock,
  \emph{Millimeter Wave Wireless Communications}.\hskip 1em plus 0.5em minus
  0.4em\relax Pearson Education, 2014.

\bibitem{11ad}
\BIBentryALTinterwordspacing
{IEEE 802.11ad}, ``{IEEE} 802.11ad standard draft {D0.1}.'' [Online].
  Available: \url{www.ieee802.org/11/Reports/tgad update.htm}
\BIBentrySTDinterwordspacing

\bibitem{Rappaport2013a}
T.~Rappaport, S.~Sun, R.~Mayzus, H.~Zhao, Y.~Azar, K.~Wang, G.~Wong, J.~Schulz,
  M.~Samimi, and F.~Gutierrez, ``Millimeter wave mobile communications for {5G}
  cellular: It will work!'' \emph{IEEE Access}, vol.~1, pp. 335--349, May 2013.

\bibitem{Roh2014}
W.~Roh, J.-Y. Seol, J.~Park, B.~Lee, J.~Lee, Y.~Kim, J.~Cho, K.~Cheun, and
  F.~Aryanfar, ``Millimeter-wave beamforming as an enabling technology for {5G}
  cellular communications: theoretical feasibility and prototype results,''
  \emph{IEEE Communications Magazine}, vol.~52, no.~2, pp. 106--113, February
  2014.

\bibitem{Alkhateeb2014d}
A.~Alkhateeb, J.~Mo, N.~Gonzalez-Prelcic, and R.~Heath, ``{MIMO} precoding and
  combining solutions for millimeter-wave systems,'' \emph{IEEE Communications
  Magazine,}, vol.~52, no.~12, pp. 122--131, Dec. 2014.

\bibitem{Heath_JSTSP16}
R.~W. Heath~Jr, N.~Gonzalez-Prelcic, S.~Rangan, W.~Roh, and A.~M. Sayeed, ``An
  overview of signal processing techniques for millimeter wave {MIMO}
  systems,'' \emph{IEEE J. Sel. Top. Signal Process.}, vol.~10, no.~3, pp.
  436--453, 2016.

\bibitem{Bai_Tianyang_TWC15}
T.~Bai and R.~W. Heath, ``Coverage and rate analysis for millimeter-wave
  cellular networks,'' \emph{{IEEE} Trans. Wireless Commun.}, vol.~14, no.~2,
  pp. 1100--1114, Feb 2015.

\bibitem{Murmann_16}
\BIBentryALTinterwordspacing
B.~Murmann, ``{ADC} performance survey 1997-2016,'' 2016. [Online]. Available:
  \url{http://www.stanford.edu/~murmann/adcsurvey.html}
\BIBentrySTDinterwordspacing

\bibitem{Hong_Wonbin_COMM14}
W.~Hong, K.-H. Baek, Y.~Lee, Y.~Kim, and S.-T. Ko, ``{Study and prototyping of
  practically large-scale mmWave antenna systems for 5G cellular devices},''
  \emph{{IEEE} Commun. Mag.}, vol.~52, no.~9, pp. 63--69, September 2014.

\bibitem{Alkhateeb_COMM14}
A.~Alkhateeb, J.~Mo, N.~Gonzalez-Prelcic, and R.~W. Heath~Jr, ``{MIMO}
  precoding and combining solutions for millimeter-wave systems,'' \emph{{IEEE}
  Commun. Mag.}, vol.~52, no.~12, pp. 122--131, December 2014.

\bibitem{Ayach_TWC14}
O.~El~Ayach, S.~Rajagopal, S.~Abu-Surra, Z.~Pi, and R.~W. Heath~Jr, ``Spatially
  sparse precoding in millimeter wave {MIMO} systems,'' \emph{{IEEE} Trans.
  Wireless Commun.}, vol.~13, no.~3, pp. 1499--1513, March 2014.

\bibitem{Han2015}
S.~Han, C.-L. I, Z.~Xu, and C.~Rowell, ``Large-scale antenna systems with
  hybrid analog and digital beamforming for millimeter wave {5G},'' \emph{IEEE
  Communications Magazine}, vol.~53, no.~1, pp. 186--194, Jan. 2015.

\bibitem{Mo_Jianhua_TSP15}
J.~{Mo} and R.~W. {Heath}~Jr., ``Capacity analysis of one-bit quantized {MIMO}
  systems with transmitter channel state information,'' \emph{{IEEE} Trans.
  Signal Processing}, vol.~63, no.~20, pp. 5498--5512, Oct 2015.

\bibitem{Zhang2005a}
X.~Zhang, A.~Molisch, and S.~Kung, ``Variable-phase-shift-based {RF}-baseband
  codesign for {MIMO} antenna selection,'' \emph{IEEE Transactions on Signal
  Processing}, vol.~53, no.~11, pp. 4091--4103, Nov. 2005.

\bibitem{Venkateswaran2010}
V.~Venkateswaran and A.~van~der Veen, ``Analog beamforming in {MIMO}
  communications with phase shift networks and online channel estimation,''
  \emph{IEEE Transactions on Signal Processing}, vol.~58, no.~8, pp.
  4131--4143, Aug. 2010.

\bibitem{Alkhateeb2013}
A.~Alkhateeb, O.~El~Ayach, G.~Leus, and R.~Heath, ``Hybrid precoding for
  millimeter wave cellular systems with partial channel knowledge,'' in
  \emph{Proc. of Information Theory and Applications Workshop (ITA)}, Feb 2013,
  pp. 1--5.

\bibitem{Alkhateeb2014}
------, ``Channel estimation and hybrid precoding for millimeter wave cellular
  systems,'' \emph{IEEE Journal of Selected Topics in Signal Processing},
  vol.~8, no.~5, pp. 831--846, Oct. 2014.

\bibitem{Rusu_ICC15}
C.~Rusu, R.~M茅ndez-Rial, N.~Gonz谩lez-Prelcicy, and R.~W. Heath, ``Low
  complexity hybrid sparse precoding and combining in millimeter wave mimo
  systems,'' in \emph{Proceeedings of the 2015 IEEE International Conference on
  Communications (ICC)}, June 2015, pp. 1340--1345.

\bibitem{Chen2015}
C.-E. Chen, ``An iterative hybrid transceiver design algorithm for millimeter
  wave {MIMO} systems,'' \emph{IEEE Wireless Communications Letters}, vol.~4,
  no.~3, pp. 285--288, June 2015.

\bibitem{Bogale2014}
T.~Bogale and L.~B. Le, ``Beamforming for multiuser massive {MIMO} systems:
  Digital versus hybrid analog-digital,'' in \emph{IEEE Global Communications
  Conference (GLOBECOM)}, Dec. 2014, pp. 4066--4071.

\bibitem{Sohrabi2015}
F.~Sohrabi and W.~Yu, ``Hybrid digital and analog beamforming design for
  large-scale {MIMO} systems,'' in \emph{Proc. of the IEEE International Conf.
  on Acoustics, Speech and Signal Processing (ICASSP), Brisbane, Australia},
  Apr. 2015.

\bibitem{Liang2014}
L.~Liang, W.~Xu, and X.~Dong, ``Low-complexity hybrid precoding in massive
  multiuser {MIMO} systems,'' \emph{IEEE Wireless Communications Letters},
  vol.~3, no.~6, pp. 653--656, Dec 2014.

\bibitem{Venkateswaran_TSP10}
V.~Venkateswaran and A.-J. van~der Veen, ``Analog beamforming in {MIMO}
  communications with phase shift networks and online channel estimation,''
  \emph{{IEEE} Trans. Signal Processing}, vol.~58, no.~8, pp. 4131--4143, 2010.

\bibitem{Walden_JSAC99}
R.~Walden, ``Analog-to-digital converter survey and analysis,'' \emph{{IEEE} J.
  Select. Areas Commun.}, vol.~17, no.~4, pp. 539--550, 1999.

\bibitem{Le_Bin_SPM05}
B.~Le, T.~Rondeau, J.~Reed, and C.~Bostian, ``Analog-to-digital converters,''
  \emph{{IEEE} Signal Processing Mag.}, vol.~22, no.~6, pp. 69--77, 2005.

\bibitem{Mezghani_ISIT07}
A.~Mezghani and J.~Nossek, ``On ultra-wideband {MIMO} systems with 1-bit
  quantized outputs: Performance analysis and input optimization,'' in
  \emph{Proceeedings of IEEE International Symposium on Information Theory},
  2007, pp. 1286--1289.

\bibitem{Mezghani_WSA07}
A.~Mezghani, M.-S. Khoufi, and J.~A. Nossek, ``{A modified MMSE receiver for
  quantized MIMO systems},'' \emph{Proceedings of the ITG/IEEE WSA, Vienna,
  Austria}, 2007.

\bibitem{Mezghani_ISIT12}
A.~Mezghani and J.~Nossek, ``Capacity lower bound of {MIMO} channels with
  output quantization and correlated noise,'' in \emph{Proceeedings of IEEE
  International Symposium on Information Theory}, 2012.

\bibitem{Zhang_Wenyi_TCOM12}
W.~Zhang, ``A general framework for transmission with transceiver distortion
  and some applications,'' \emph{{IEEE} Trans. Commun.}, vol.~60, no.~2, pp.
  384--399, February 2012.

\bibitem{Mo_Jianhua_Asilomar14}
J.~Mo, P.~Schniter, N.~G. Prelcic, and R.~W. Heath~Jr., ``Channel estimation in
  millimeter wave {MIMO} systems with one-bit quantization,'' in
  \emph{Proceeedings of the 2014 48th Asilomar Conference on Signals, Systems
  and Computers}, Nov 2014, pp. 957--961.

\bibitem{Bai_Qing_ETT15}
Q.~Bai and J.~Nossek, ``Energy efficiency maximization for {5G} multi-antenna
  receivers,'' \emph{Transactions on Emerging Telecommunications Technologies},
  vol.~26, no.~1, pp. 3--14, 2015.

\bibitem{Orhan_ITA15}
O.~{Orhan}, E.~{Erkip}, and S.~{Rangan}, ``{Low Power Analog-to-Digital
  Conversion in Millimeter Wave Systems: Impact of Resolution and Bandwidth on
  Performance},'' in \emph{Proc. of Information Theory and Applications (ITA)
  Workshop}, 2015.

\bibitem{Wang_Shengchu_TWC15}
S.~Wang, Y.~Li, and J.~Wang, ``Multiuser detection in massive spatial
  modulation {MIMO} with low-resolution {ADCs},'' \emph{{IEEE} Trans. Wireless
  Commun.}, vol.~14, no.~4, pp. 2156--2168, April 2015.

\bibitem{Jacobsson_arxiv15}
S.~Jacobsson, G.~Durisi, M.~Coldrey, U.~Gustavsson, and C.~Studer, ``One-bit
  massive {MIMO}: Channel estimation and high-order modulations,'' \emph{arXiv
  preprint arXiv:1504.04540}, 2015.

\bibitem{Singh_TCOM09}
J.~Singh, O.~Dabeer, and U.~Madhow, ``On the limits of communication with
  low-precision analog-to-digital conversion at the receiver,'' \emph{{IEEE}
  Trans. Commun.}, vol.~57, no.~12, pp. 3629--3639, 2009.

\bibitem{Mo_Jianhua_ITA14}
J.~Mo and R.~W. Heath~Jr., ``High {SNR} capacity of millimeter wave {MIMO}
  systems with one-bit quantization,'' in \emph{Proceeedings of Information
  Theory and Applications Workshop (ITA), 2014}, Feb 2014, pp. 1--5.

\bibitem{Choi_TCOM16}
J.~Choi, J.~Mo, and R.~W. Heath~Jr., ``Near maximum-likelihood detector and
  channel estimator for uplink multiuser massive {MIMO} systems with one-bit
  {ADCs},'' \emph{{IEEE} Trans. Commun.}, vol.~64, no.~5, pp. 2005--2018, May
  2016.

\bibitem{Mollen_arxiv16}
\BIBentryALTinterwordspacing
C.~Mollen, J.~Choi, E.~G. Larsson, and R.~W.~H. Jr., ``Performance of the
  wideband massive uplink {MIMO} with one-bit {ADCs},'' \emph{CoRR}, vol.
  abs/1602.07364, 2016. [Online]. Available:
  \url{http://arxiv.org/abs/1602.07364}
\BIBentrySTDinterwordspacing

\bibitem{Jacobsson_arxiv16}
\BIBentryALTinterwordspacing
S.~Jacobsson, G.~Durisi, M.~Coldrey, U.~Gustavsson, and C.~Studer, ``Massive
  {MIMO} with low-resolution {ADCs},'' \emph{CoRR}, vol. abs/1602.01139, 2016.
  [Online]. Available: \url{http://arxiv.org/abs/1602.01139}
\BIBentrySTDinterwordspacing

\bibitem{Studer_TCOM16}
C.~Studer and G.~Durisi, ``Quantized massive {MU-MIMO-OFDM} uplink,''
  \emph{IEEE Trans. Commun.}, vol.~64, no.~6, pp. 2387--2399, June 2016.

\bibitem{Murmann_FTFC13}
B.~Murmann, ``Energy limits in {A/D} converters,'' in \emph{Faible Tension
  Faible Consommation (FTFC), 2013 IEEE}, June 2013, pp. 1--4.

\bibitem{Ghauch2015}
H.~Ghauch, T.~Kim, M.~Bengtsson, and M.~Skoglund, ``Subspace estimation and
  decomposition for large millimeter-wave mimo systems,'' \emph{submitted to
  IEEE Journal of Selected Topics in Signal Processing, arXiv preprint
  arXiv:1507.00287}, 2015.

\bibitem{Mo_Jianhua_arxiv16b}
J.~Mo, P.~Schniter, and R.~W. Heath~Jr, ``Channel estimation in broadband
  millimeter wave {MIMO} systems with few-bit {ADCs},'' \emph{arXiv preprint
  arXiv:1610.02735}, 2016.

\bibitem{Kamilov_TSP12}
U.~Kamilov, V.~Goyal, and S.~Rangan, ``Message-passing de-quantization with
  applications to compressed sensing,'' \emph{{IEEE} Trans. Signal Processing},
  vol.~60, no.~12, pp. 6270--6281, Dec 2012.

\bibitem{Molisch_CL04}
A.~Molisch and X.~Zhang, ``{FFT-based hybrid antenna selection schemes for
  spatially correlated MIMO channels},'' \emph{{IEEE} Commun. Lett.}, vol.~8,
  no.~1, pp. 36--38, Jan 2004.

\bibitem{Rao_Math79}
C.~R. Rao, ``Separation theorems for singular values of matrices and their
  applications in multivariate analysis,'' \emph{Journal of Multivariate
  Analysis}, vol.~9, no.~3, pp. 362--377, 1979.

\bibitem{Palomar_TSP03}
D.~P. Palomar, J.~M. Cioffi, and M.~A. Lagunas, ``{Joint Tx-Rx beamforming
  design for multicarrier MIMO channels: a unified framework for convex
  optimization},'' \emph{{IEEE} Trans. Signal Processing}, vol.~51, no.~9, pp.
  2381--2401, Sept 2003.

\bibitem{Tropp_IT05}
J.~A. Tropp, I.~S. Dhillon, R.~W. Heath, and T.~Strohmer, ``Designing
  structured tight frames via an alternating projection method,'' \emph{{IEEE}
  Trans. Inform. Theory}, vol.~51, no.~1, pp. 188--209, 2005.

\bibitem{Sayeed_TSP02}
A.~Sayeed, ``Deconstructing multiantenna fading channels,'' \emph{{IEEE} Trans.
  Signal Processing}, vol.~50, no.~10, pp. 2563--2579, Oct 2002.

\bibitem{Cover_Book12}
T.~M. Cover and J.~A. Thomas, \emph{Elements of information theory}.\hskip 1em
  plus 0.5em minus 0.4em\relax John Wiley \& Sons, 2012.

\bibitem{Proakis_Book08}
J.~G. Proakis, ``Digital communications.'' \emph{McGraw-Hill, New York}, 2008.

\bibitem{Fletcher_JSTSP07}
A.~Fletcher, S.~Rangan, V.~Goyal, and K.~Ramchandran, ``Robust predictive
  quantization: Analysis and design via convex optimization,'' vol.~1, no.~4,
  pp. 618--632, Dec 2007.

\bibitem{Max_IRE60}
J.~Max, ``Quantizing for minimum distortion,'' \emph{IRE Transactions on
  Information Theory}, vol.~6, no.~1, pp. 7--12, 1960.

\bibitem{Lloyd_IT82}
S.~Lloyd, ``Least squares quantization in {PCM},'' \emph{{IEEE} Trans. Inform.
  Theory}, vol.~28, no.~2, pp. 129--137, Mar 1982.

\bibitem{Gersho_Book12}
A.~Gersho and R.~M. Gray, \emph{Vector quantization and signal
  compression}.\hskip 1em plus 0.5em minus 0.4em\relax Springer Science \&
  Business Media, 2012, vol. 159.

\bibitem{Akdeniz_JSAC14}
M.~Akdeniz, Y.~Liu, M.~Samimi, S.~Sun, S.~Rangan, T.~Rappaport, and E.~Erkip,
  ``Millimeter wave channel modeling and cellular capacity evaluation,''
  \emph{{IEEE} J. Select. Areas Commun.}, vol.~32, no.~6, pp. 1164--1179, June
  2014.

\bibitem{Rappaport_TCOM15}
T.~Rappaport, G.~Maccartney, M.~Samimi, and S.~Sun, ``Wideband millimeter-wave
  propagation measurements and channel models for future wireless communication
  system design,'' \emph{{IEEE} Trans. Commun.}, vol.~63, no.~9, pp.
  3029--3056, Sept 2015.

\bibitem{WP_5G_Channel_Model}
``{5G} channel model for bands up to100 {GHz},'' Tech. Rep., 2015.

\bibitem{Mendez-Rial_Access16}
R.~Méndez-Rial, C.~Rusu, N.~González-Prelcic, A.~Alkhateeb, and R.~W.~H. Jr,
  ``Hybrid {MIMO} architectures for millimeter wave communications: Phase
  shifters or switches?'' \emph{IEEE Access}, vol.~4, pp. 247--267, 2016.

\bibitem{Lin_Yu-Hsuan_IMS16}
Y.-H. Lin and H.~Wang, ``{A low phase and gain error passive phase shifter in
  90 nm CMOS for 60 GHz phase array system application},'' in \emph{Proceedings
  of the 2016 IEEE MTT-S International Microwave Symposium (IMS)}, May 2016,
  pp. 1--4.

\bibitem{Ali_JSSC14}
A.~M.~A. Ali, H.~Dinc, P.~Bhoraskar, C.~Dillon, S.~Puckett, B.~Gray, C.~Speir,
  J.~Lanford, J.~Brunsilius, P.~R. Derounian, B.~Jeffries, U.~Mehta, M.~McShea,
  and R.~Stop, ``{A 14 Bit 1 GS/s RF Sampling Pipelined ADC With Background
  Calibration},'' \emph{IEEE J. Solid-State Circuits}, vol.~49, no.~12, pp.
  2857--2867, Dec 2014.

\bibitem{Dabeer_SPAWC06}
O.~Dabeer, J.~Singh, and U.~Madhow, ``On the limits of communication
  performance with one-bit analog-to-digital conversion,'' in \emph{Proceedings
  of the IEEE 7th Workshop on Signal Processing Advances in Wireless
  Communications}, 2006, pp. 1--5.

\end{thebibliography}
\end{document}